\documentclass[12pt]{article}
\usepackage{amsfonts}

\usepackage{graphicx}
\usepackage{amsmath,amssymb}
\usepackage{bm}

\usepackage{color}
\usepackage{amsmath,amssymb}
\usepackage{color}
\usepackage{hyperref,bookmark}
\usepackage[english]{babel}

\usepackage{soul}
\usepackage{multirow}
\usepackage{float}
\usepackage{ulem}
\usepackage{array}
\usepackage{comment}

\newcommand{\Real}{\mathbb{R}}
\newcommand{\E}{\mathbb{E}}
\newcommand{\Rlogo}{\protect\includegraphics[height=1.8ex,keepaspectratio]{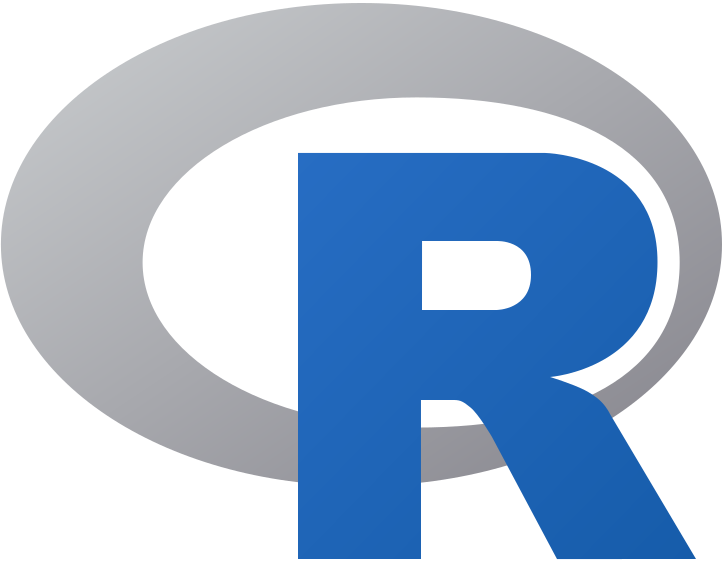}}

\newtheorem{proposition}{Proposition}
\newtheorem{example}{Example}
\newtheorem{definition}{Definition}
\newtheorem{remark}{Remark}
 \newtheorem{proof}{Proof}

\begin{document}

\title{A new Bayesian discrepancy measure}
\date{}

\author{
\small Francesco Bertolino$^{\rm a}$, Mara Manca$^{\rm a}$, Monica Musio$^{\rm a}$, Walter Racugno$^{\rm a}$ and Laura Ventura$^{\rm b}$ \\
\footnotesize{$^{\rm a}$ {\em{Dipartimento di Matematica ed Informatica, Universit\`a di Cagliari}}} \\
\footnotesize{$^{\rm b}$ {\em{Dipartimento di Scienze Statistiche, Universit\`a di Padova}}}}

\maketitle

\begin{abstract}
The aim of this article is to make a contribution to the Bayesian procedure of testing precise  hypotheses for parametric models. For this purpose, we define the Bayesian Discrepancy Measure that allows one to evaluate the suitability of a given hypothesis with respect to the available information (prior law and data). To summarise this information, the posterior median is employed, allowing a simple assessment of the discrepancy with a fixed hypothesis. The Bayesian Discrepancy Measure assesses the compatibility of a single hypothesis with the observed data, as opposed to the more common comparative approach where a hypothesis is rejected in favour of a competing hypothesis.
The proposed measure of evidence has properties of consistency and invariance. After presenting the definition of the measure for a parameter of interest, both in the absence and in the presence of nuisance parameters, we illustrate some examples showing its conceptual and interpretative simplicity. Finally, we compare the BDT with the Full Bayesian Significance Test, a well-known Bayesian testing procedure for sharp hypotheses.\\
% that allow for easy handling of complex case studies. \\

\noindent
{\bf Keywords:} \emph{Bayesian test, Evidence, Precise hypothesis, Significance test, Full Bayesian Significance Test}
\end{abstract}

%\begin{keyword}[class=MSC]
%\kwd[Primary ]{62F15}\kwd{62F03}
%\kwd[; Secondary ]{62A}\kwd{62C10}\input{BDM-BA-revisione.tex}
%
%\end{keyword}
%
%
%\begin{keyword}
%\kwd{Bayesian test}
%\kwd{Evidence}
%\kwd{Sharp hypothesis}
%\kwd{Significance test}
%\end{keyword}
%
%\end{frontmatter}

\section{Introduction}

%\begin{acks}[Acknowledgments]
%And this is an acknowledgements section with a heading that was produced by the
%$\backslash$section* command. Thank you all for helping me writing this
%\LaTeX\ sample file.
%\end{acks}

\label{sec1}

D. V. Lindley in \cite{lind65} (preface page xi)  stated that
 \begin{quote}``\textit{
 \dots hypothesis testing looms large in standard statistical
 practice,  yet scarcely appears  as such in the Bayesian literature}.''
 \end{quote}

\noindent Since then things have changed and in the last sixty years there have been several attempts to build a measure of evidence that covers, in a Bayesian context, the role that the \textit{p-value} has played in the frequentist setting. A prominent example is the decision test based on the Bayes Factor and its extensions (see, for instance, \cite{ber85}). 
  
As an alternative to the Bayes Factor, another Bayesian evidence measure is provided in \cite{perstern99} upon which the testing procedure Full Bayesian Signicance Test (FBST) is based.  
%Although the FBST has been successfully applied to several relevant problems of statistical inference (see \cite{persternwech08}), it lacks of invariance under reparametrizations, at least in its first formulation.
For a recent survey on the FBST see \cite{perstern20}.

The main aim of this paper is to give a contribution to the  testing  procedure of precise hypotheses. In particular, the proposed Bayesian measure of evidence, called Bayesian Discrepancy Measure (BDM), gives an absolute evaluation of a hypothesis $H$ in light of prior knowledge about the parameter and observed data.  The proposed measure of evidence has the desired properties of invariance under reparametrization and consistency for large samples.

Our starting point is the idea that a hypothesis may be more or less supported
by  the available evidence contained in the posterior distribution. 

We do not adopt the hypothesis testing approach for which there is no test that can lead to the rejection of a hypothesis except by comparing it with another hypothesis (Neyman-Pearson  in the frequentist perspective, Bayes factor in the Bayesian one), but rather the approach proposed by Fisher (see \cite{christensen2005testing} and  \cite{denis2004}). 
Reference is made to a precise hypothesis $H$ and no alternative is considered against  such hypothesis.  
In this view different hypotheses made by several experts can be evaluated  and  using the information coming from the same data, some can be accepted  others not.  
In this respect, in a broad sense, we can say that we return to Fisher's original idea of pure significance according to which
 ``\textit{Every experiment may be said to exist only in order to give the facts a chance of disproving the null hypothesis}'' (\cite{fisher25}). Notice that, since the Bayesian Discrepancy Test does not require any alternative hypothesis to be specified, the Jeffreys-Lindley paradox cannot arise unlike with the Bayes Factor (see \cite{lind57}).

The structure of the paper is as follows. In Section 2 the definition of the proposed index is presented for a scalar parameter of interest, both in the absence or presence of nuisance parameters. In Section 3 different illustrative examples are discussed, involving one or two independent populations. Finally, in Section 4 we make a comparison between the Bayesian Discrepancy Test and the Full Bayesian Significance Test which is based on the $e$-value, a well-known Bayesian evidence index used to test sharp hypotheses.  The last section contains conclusions and directions for further research. 
%---------------------------------------------------------------------------------------------------------------

%----------------------------------------------- Section 2 -----------------------------------------------------
\section{The Bayesian Discrepancy Measure}\label{sec2}
Let   $(\mathcal{X}, \mathcal{P}^{X}_{\boldsymbol{\theta}}, \boldsymbol{\Theta})$ be a parametric statistical model
where $X \in \mathcal{X} \subset \Real^k$,  $\mathcal{P}^{X}_{\boldsymbol{\theta}}=\{f(x \vert\boldsymbol{\theta})\ \vert \ \boldsymbol{\theta} \in \boldsymbol{\Theta}\}$ is a class of probability distributions (Lebesgue integrable)  defined on $\mathcal{X}$, depending  on an unknown  vector of continuous  parameters $\boldsymbol{\theta} \in  \boldsymbol{\Theta}$, an  open  subset of $\Real^p$. Assume that
\begin{itemize}
    \item [(a)] the  model  is identifiable; 
    
    \item [(b)] 
  $f(x \vert\boldsymbol{\theta})$ have support not depending on $\boldsymbol{\theta}$, $\forall \ \boldsymbol{\theta} \in \boldsymbol{\Theta}$;
    \item [(c)] the operations of integration and differentiation with respect to $\boldsymbol{\theta}$ can be exchanged. 
\end{itemize}

We assume  a prior probability density $g_0(\boldsymbol{\theta})$  following Cromwell's Rule which states that ``\textit{it is inadvisable to attach probabilites of zero to uncertain events, for if the prior probability is zero so is the posterior, whatever be the data. A probability of one is equally
dangerous because then the probability of the complementary event will be zero}'' (see Section 6.2 in \cite{lindley1991}). 

First we discuss the case of a scalar parameter.  Then we discuss the case of a scalar parameter of interest in the presence of nuisance parameters.

%%%%%% SUBSECTION 1 %%%%%%
\subsection{The Baye\-sian Discrepancy Measure for a scalar parameter}\label{subsec2_univariate}
In this section we assume that $k=p=1$.

Given an \textit{iid} random sample $\boldsymbol{x}=(x_1,\ldots,x_n)$ from $\mathcal{P}^{X}_{\theta}$, let $L(\theta\vert \boldsymbol{x})$ be the corresponding  likelihood function based on data $\boldsymbol x$ and let $g_0(\theta)$ be the prior distribution on $\Theta \subseteq \Real$. The posterior probability density for $\theta$ given $\boldsymbol x$ is then
$$
g_1(\theta \vert \boldsymbol x) \propto g_0(\theta) \, L(\theta\vert\boldsymbol x).
$$

\noindent
Moreover, given the posterior distribution function $G_1(\theta \vert \boldsymbol{x})$, the posterior median is any real number $m_1$ which satisfies the inequalities $G_1(m_1\vert \boldsymbol{x}) \geq \frac{1}{2}$ and $G_1^-(m_1 \vert \boldsymbol{x}) \leq \frac{1}{2}$, where $\displaystyle G_1^-(m_1 \vert \boldsymbol{x})=\lim_{\theta \uparrow m_1} G_1(\theta \vert \boldsymbol{x})$. In the case in which $G_1$ is continuous and strictly increasing we have $m_1= G_1^{-1}(\frac{1}{2} \vert \boldsymbol{x})$.

\noindent
We are interested in testing the precise hypothesis
\begin{equation}
\label{eq:1}
H: \theta = \theta_H.
\end{equation}

\noindent
In order to measure the discrepancy of the hypothesis (\ref{eq:1})  w.r.t.\  the posterior distribution, in the case $\Theta=\Real$, we consider the following two intervals:

\begin{enumerate}
\item the \textit{discrepancy interval}
\begin{eqnarray}
I_H = \left\{
\begin{array}{ll}
(m_1,\theta_H) & \text{\  \ if \ } \quad  m_1 < \theta_H \\
\{m_1\} & \text{\ \ if \ } \quad  m_1 = \theta_H, \\
(\theta_H,m_1) & \text{\ \ if\ \ } \quad  m_1 > \theta_H \\
\end{array}
\right.
\end{eqnarray}
\item
the  \textit{external  interval}
\begin{eqnarray}
I_E = \left\{
\begin{array}{ll}
(\theta_H,+\infty) & \text{\ \ if\ } \quad  m_1 < \theta_H \\
%\R & \text{if} \quad  m_1 = \theta_H, \\
(-\infty,\theta_H) & \text{\ \ if\ } \quad   \theta_H < m_1. \\
\end{array}
\right.
\end{eqnarray}
\end{enumerate}

When $m_1=\theta_H$, the external interval $I_E$ can be $(-\infty, m_1)$ or $(m_1,+\infty)$.
Note that, by construction, $\mathbb{P}(I_H \cup I_E)=\frac{1}{2}$ (see Figure \ref{fig1}).  If  the support of the posterior is a subset of $\Real$, the intervals $I_H$ and $I_E$ can be defined consequently. 

\begin{figure}
\begin{center}
\includegraphics[scale=0.8]{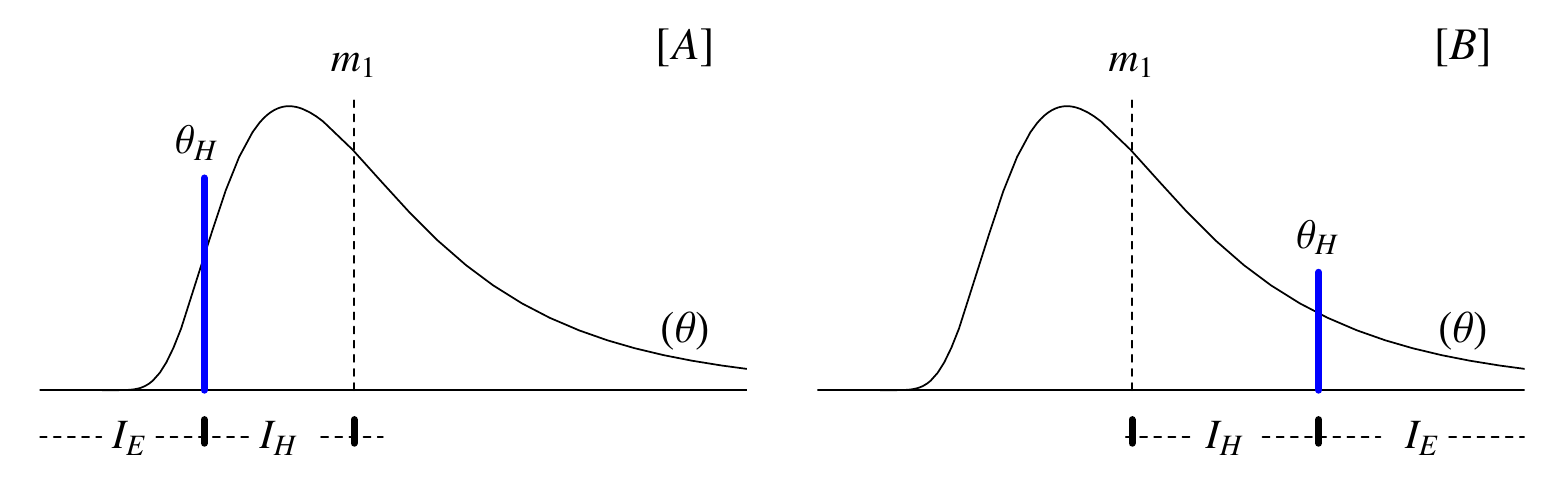}
\vspace{-0cm}
\caption{{\small Posterior density $g_1(\theta \vert \boldsymbol{x})$, the corresponding discrepancy interval $I_H$ and external interval  $I_E$ when $\theta_H < m_1$ ([A]) and $\theta_H >  m_1$ ([B]). }}
\label{fig1}
\end{center}
\end{figure}

\begin{definition}
Given the posterior distribution function 
$G_1(\theta \vert \boldsymbol{x})$,  we define the  Baye\-sian Discrepancy Measure of the hypothesis $H$  as
\begin{eqnarray}
\label{formula:4}
\delta_H = 2 \,  \mathbb{P}(\theta \in I_H \vert \boldsymbol{x}) = 2 \int_{I_H} d G_1(\theta \vert \boldsymbol{x}).
\end{eqnarray}
\end{definition}
\noindent
The measure can be also computed by means of the external interval as
\begin{eqnarray}
\delta_H = 1 - 2 \,  \mathbb{P}(\theta \in I_E  \vert \boldsymbol{x}) = 1 - 2 \int_{I_E} d G_1(\theta \vert \boldsymbol{x}),
\label{formula:est}
\end{eqnarray}
which can also be written as
%\begin{eqnarray}
%\delta_H = 1 - 2 \min\{G_1(\theta_H \vert \boldsymbol{x}),1-G_1^{-}(\theta_H \vert \boldsymbol{x})\}
%\label{formula:tail}
%\end{eqnarray}
%that, in the absolutely continuous case, reduces to
\begin{eqnarray}
\delta_H = 1 - 2 \min \{G_1(\theta_H \vert \boldsymbol{x}),1-G_1^-(\theta_H \vert \boldsymbol{x})\},
\label{formula:tail}
\end{eqnarray}
where $\displaystyle G_1^-(\theta_H \vert \boldsymbol{x})=\lim_{\theta \uparrow \theta_H} G_1(\theta \vert \boldsymbol{x})$.
In the absolutely continuous case, this simplifies to 
\begin{eqnarray}
\delta_H = 1 - 2 \min \{G_1(\theta_H \vert \boldsymbol{x}),1-G_1(\theta_H \vert \boldsymbol{x})\}.
\label{formula:tail1}
\end{eqnarray}
Formulations (\ref{formula:tail}) and (\ref{formula:tail1}) have the advantage of  not involving the  posterior median in the integral computation. 
Furthermore, one can interpret the quantity $\min \{G_1(\theta_H \vert \boldsymbol{x}),1-G_1(\theta_H \vert \boldsymbol{x})\}$ as the posterior probability of a ``tail" event concerning only the precise hypothesis $H$. Doubling this ``tail" probability, related to the precise hypothesis $H$, one gets a posterior probability assessment about how ``central" the hypothesis $H$ is, and hence how it is supported by the prior and the data.

It is important to highlight that the hypothesis $H$ induces the following partition
\begin{equation}
    \big\{ \Theta_a = (-\infty,  \theta_H),\ \Theta_H = \{\theta_H\},\ \Theta_b = (\theta_H , \infty)\big\}
\end{equation}
of the parameter space $\Theta$. Then formulations (\ref{formula:tail}) and (\ref{formula:tail1}) can be equivalently expressed as 
%---------------------------------------------------------------------- m.d.  DELTA
  \begin{equation}\label{misuraDELTA-03}
  \delta_H   = \; 1  - 2 \cdot    \min_{a,b}
  \displaystyle \big\{ {\mathbb{P}}(\theta \in \Theta_a \vert \boldsymbol{x}) \, , \, {\mathbb{P}}(\theta \in \Theta_b \vert \boldsymbol{x}) \big\}
  \, .
  \end{equation}
%-------------------------------------------------------------------------------------------------------------------------------------------------
The last formula can be naturally extended to the case where, besides the scalar parameter of interest, nuisance parameters are also present. This issue will be developed in Section \ref{subsec2_general}.

As pointed out before, the further $\theta_H$ is from the posterior median $m_1$ of the distribution function $G_1(\theta \vert \boldsymbol{x})$, the closer $\delta_H$ is to $1$. It can then be said that $H$ does \textit{not conform} to $G_1(\theta \vert \boldsymbol{x})$. On the contrary, the smaller $\delta_H$ the stronger is the evidence in favor of $H$. Following this idea, we can define a general testing procedure by choosing a certain threshold to establish how large the measure must be, before we can state that $H$ does not conform to the posterior distribution function.

\begin{definition}
The Bayesian Discrepancy Test (BDT) is the procedure based on the
Baye\-sian Discrepancy Measure (BDM) that rejects the hypothesis $H$ when $\delta_H$ is higher than some critical value $\omega \in \{0.95, 0.99, 0.995, 0.999, \dots\}$. 
\end{definition}

\noindent
As for all measures of evidence (Bayesian or frequentist), the chosen value for $\omega$ inevitably has a character of subjectivity. For a more detailed discussion on the threshold choice see the remark in Example \ref{primoESEMPIO}.

\begin{proposition}
\label{prop}
The following properties apply to the BDM, for a scalar parameter $\theta$: 
\begin{itemize}
\item[(i)]
$\delta_H$ always exists and, by construction, $\delta_H \in [0, 1)$;
\item[(ii)]
 $\delta_H$ is invariant under invertible monotonic transformations of the parameter $\theta$;
\item[(iii)]
  if $\theta^*$ is the true value of the parameter and $\theta_H=\theta^*$, then  $\delta_{H} \sim Unif(\cdot \vert 0,1)$,  for all sample sizes $n$; if  $\theta_H \neq \theta^*$,
 then $\delta_{H} \, {\mathop{\to}\limits^{\textit{p}}}\, 1$ (consistency property).  
\end{itemize}
\end{proposition}

\begin{proof}
\begin{itemize}
    \item[\textit{(i)}] The first property follows immediately from the fact that in (\ref{formula:4}) the posterior probability $\mathbb{P}(\theta \in I_H \vert \boldsymbol{x}) \in \Big[0,\frac{1}{2}\Big]$. 
    
    \item[\textit{(ii)}] Let $\lambda=\lambda(\theta)$ be an invertible monotonic transformation of the parameter $\theta$ and let $K_1$ be the cumulative distribution function of the parameter $\lambda$. We denote with $\lambda_H=\lambda(\theta_H)$ and we notice that $m'_1=\lambda(m_1)$ thanks to the monotonic invariance of the median. Suppose, for simplicity, that $\theta_H>m_1$. Then    
    \begin{equation*}
        \delta_H=2\ \int_{m_1}^{\theta_H} dG_1(\theta \vert \boldsymbol{x})\ = 2\ \Big\vert \int_{m'_1}^{\lambda_H} dK_1(\lambda \vert \boldsymbol{x})\Big\vert.
    \end{equation*}
    Therefore, the invariance of the BDM follows immediately from the invariance of the median under invertible monotonic transformations.
    Notice  that if instead of the median $m_1$ we consider, for example, the posterior mean  $E(\theta \vert  \boldsymbol{x} )$, which is not invariant under invertible monotonic reparametrizations, the property will not hold in general. Moreover, $E(\theta \vert  \boldsymbol{x} )$ for some models may not even exist.
    
    \item[\textit{(iii)}] Let us examine the first part of the statement for which $\theta_H=\theta^*$. Suppose that  $\theta_H<m_1$. The BDM is defined as
    \begin{align}
    \begin{split}
      \delta_H &= 2 \int_{\theta_H}^{m_1} dG_1(\theta \vert \boldsymbol{x}) \\
      &=1 - 2 \int_{-\infty}^{\theta_H} dG_1(\theta \vert \boldsymbol{x}) \\
      &= 1-2\ G_1(\theta_H \vert \boldsymbol{x}).
    \end{split}
    \label{B}
    \end{align}
    Using the integral transform and the fact that we have supposed $\theta_H<m_1$, we easly find that
    \begin{equation*}
        G_1(\theta_H \vert \boldsymbol{x}) = \int_{-\infty}^{\theta_H} {\rm d}G_1(\theta \vert \boldsymbol{x})  = W \sim Unif\big(\cdot\big \vert 0,\textstyle{\frac{1}{2}}\big).
    \end{equation*}
    Then, since  $\delta_H = 1 - 2 W$,  we find  $\delta_H\sim Unif(\cdot \vert 0,1).$  A similar proof holds for  $\theta_H>m_1$.
     If, instead, $\theta_H \neq \theta^*$ and  $n \rightarrow \infty$, under suitable regularity conditions (see for instance Section 5.3.2, p. 287 in \cite{bersmith94}) it is  well known that $g_1(\theta \vert \boldsymbol{x})$ is concentrated around $\theta^*$. In particular, the posterior median $m_1$ converges in probability to $\theta^*$. Again, suppose for instance that $\theta_H < \theta^*$, then $\lim_{n \rightarrow \infty} \delta_H= 2 \lim_{n \rightarrow \infty} \int_{\theta_H}^{m_1} dG_1(\theta \vert \boldsymbol{x})= 2 \cdot \frac{1}{2}= 1.$
\end{itemize}
\end{proof}

%%%%%% SUBSECTION 2 %%%%%%
\subsection{The Baye\-sian Discrepancy Measure in presence of nuisance parameters}\label{subsec2_general}
\label{multiparametric} 
Suppose that $p \geq 2$ and $k \geq 1$. Let $\varphi= \varphi(\boldsymbol{\theta})$ be a scalar parameter of interest, where $\varphi: \boldsymbol{\Theta} \to \Phi \subseteq \Real$. Let us further consider a bijective reparametrization $\boldsymbol{\theta} \Leftrightarrow (\varphi, \boldsymbol{\zeta})$, where $\boldsymbol{\zeta} \in \boldsymbol{Z} \subseteq \Real^{p-1}$ denotes an arbitrary nuisance parameter, which is determined on the basis of analytical convenience (note that the value of the evidence measure is invariant with respect to the choice of the nuisance parameter).   
We consider hypotheses that can be expressed in the form
\begin{equation}\label{ipotesiMULT}
 H: \varphi = \varphi_H,
\end{equation}
where $\varphi_H$ is known as it represents the hypothesis that it is of interest to evaluate. The transformation $\varphi$ must be such that, for all $\boldsymbol{\theta} \in \boldsymbol{\Theta}$ and for all $\varphi_H \in \Phi$, it can always be assessed whether $\varphi$ is strictly smaller, strictly larger or equal to $\varphi_H$ (i.e. $\varphi < \varphi_H$ either $\varphi > \varphi_H$, or $\varphi = \varphi_H$). Hypothesis \eqref{ipotesiMULT} and transformation $\varphi$ univocally identify the partition $\big\{ \boldsymbol{\Theta}_a, \, \boldsymbol{\Theta}_H, \, \boldsymbol{\Theta}_b \big\}$ of the parameter space $\boldsymbol{\Theta}$, with
\begin{equation}\label{partition}
    \begin{array}{ll}
        {\boldsymbol{\Theta}}_a &= \big\{ \boldsymbol{\theta} \in \boldsymbol{\Theta} : \varphi < \varphi_H \big\}\\
        {\boldsymbol{\Theta}}_H &= \big\{ \boldsymbol{\theta} \in \boldsymbol{\Theta} : \varphi = \varphi_H \big\}.\\
        {\boldsymbol{\Theta}}_b &= \big\{ \boldsymbol{\theta} \in \boldsymbol{\Theta} : \varphi > \varphi_H \big\}
    \end{array}
\end{equation}
We call any hypothesis of type \eqref{ipotesiMULT}, which identify a partition of the form \eqref{partition}, a \textit{partitioning hypothesis}.  It is easy to verify that many commonly used hypotheses are partitioning. In this paper we only consider hypotheses of this nature. In this setting, we express the BDM as 
 \begin{equation}\label{misuraDELTA-04}
   \begin{array}{ll} \vspace{4 pt}
   \delta_H   &=\; \displaystyle 1  -    2 \cdot \min_{a,b}
   \big\{ {\mathbb{P}}(\boldsymbol{\theta} \in \boldsymbol{\Theta}_a \vert \boldsymbol{x}) \, , \,
         {\mathbb{P}}(\boldsymbol{\theta} \in \boldsymbol{\Theta}_b \vert \boldsymbol{x}) \big\}     \\ \vspace{3 pt}
   &=\; \displaystyle   1  -    2 \cdot \int_{I_E} g_1(\boldsymbol{\theta} \vert \boldsymbol{x}) \, {\rm d}\boldsymbol{\theta},
   \,
  % . \; \lhd
 \end{array}
 \end{equation}
where  \textit{the external set}  is given by

\begin{equation}\label{insiemeEST}
  I_E \;  = \;
  \arg \min_{a,b}
  \big\{ {\mathbb{P}}(\boldsymbol{\theta} \in \boldsymbol{\Theta}_a \vert \boldsymbol{x}) \, , \,
         {\mathbb{P}}(\boldsymbol{\theta} \in \boldsymbol{\Theta}_b \vert  \boldsymbol{x}) \big\}
         \;   .  \;
         % \lhd
  \end{equation}
\noindent
In the particular scenario where the marginal posterior  
\begin{equation*}
 h_1(\varphi \vert \boldsymbol{x}) \; = \;
 \int_{\varphi( \boldsymbol{\theta})= \varphi}
 g_1(\boldsymbol{\theta} \vert \boldsymbol{x})  {\rm d}\boldsymbol{\theta} \; , \quad
 \forall \varphi \in \Phi \,  ,
\end{equation*}
of the parameter of interest $\varphi$ can be computed in a closed form, the hypothesis (\ref{ipotesiMULT}) can be easily treated using the methodologies seen in Subsection \textbf{\ref{subsec2_univariate}}, i.e. the BDM is computed by means of formula $(\ref{formula:4})$ or $(\ref{formula:est})$ applied to the marginal.

Properties reported in Proposition \ref{prop} naturally extend to the setting we just presented.  

%-------------------------------------------------------------------------------------------------------------------------

%---------------------------------------- Section 3 ----------------------------------------------------------------------
\section{Illustrative examples}\label{sec4}
The simplicity of the  BDT is highlighted by the following examples, some of which deal with cases not usually considered in the literature. Examples \ref{primoESEMPIO} and \ref{secondoESEMPIO} focus on a scalar parameter of interest, while Examples \ref{modelloGAMMA}, \ref{cv}, \ref{Mara}, \ref{bisanzioESEMPIO}, \ref{2Gamma} also contain nuisance parameters.

In all examples we have adopted a Jeffrey's prior (see \cite{yang1996} for a catalog of non-informative priors) for simplicity. However, other objective priors and, in the presence of substantive prior information, informative priors could equally be used.

\subsection{Examples of the univariate parameter case}

\begin{example}\rm{\textit{Exponential distribution\\}}
\label{primoESEMPIO}
Let  $\boldsymbol{x}=(x_1, \dots, x_n)$  be an \textit{iid} sample  of size $n$ from the exponential distribution $X \sim Exp\big(x \vert \theta^{-1} \big)$, with  $\theta \in {\Real}^+.$ We are interested in the hypothesis $H: \theta=\theta_H$.  Assuming a Jeffreys' prior for  $\theta$, i.e.  $g_0(\theta) \propto \theta^{-1}$,  the posterior distribution is  given by $g_1(\theta \vert  \boldsymbol{x}) \propto \theta^{-n-1} \exp \{- n \bar{x} \cdot \theta^{-1} \}$, with $\bar{x}$  the sample mean.

Figure \ref{fig2} shows the posterior density function as well as the  discrepancy and the  external intervals for $H:\theta = \theta_H = 2.4$ and the MLE $\bar{x} = 1.2$ for three sample sizes
   $[A]$ $n =6$, $[B]$ $n = 12$, $[C]$ $n = 24$.  In $[A]$ we have a posterior median $m_1=1.27$ and  $\delta_H =0.832$,  while in $[B]$ $m_1=1.23$ and $\delta_H =0.960$, in $[C]$ $m_1=1.22$ and $\delta_H =0.997$. In case $[A]$ we do not reject  $H$ while in $[B]$ and in $[C]$ we are led to reject the hypothesis.  %Inferential conclusions depend on the sample size. }

   \begin{figure}
    \centering
    \includegraphics[scale=0.8]{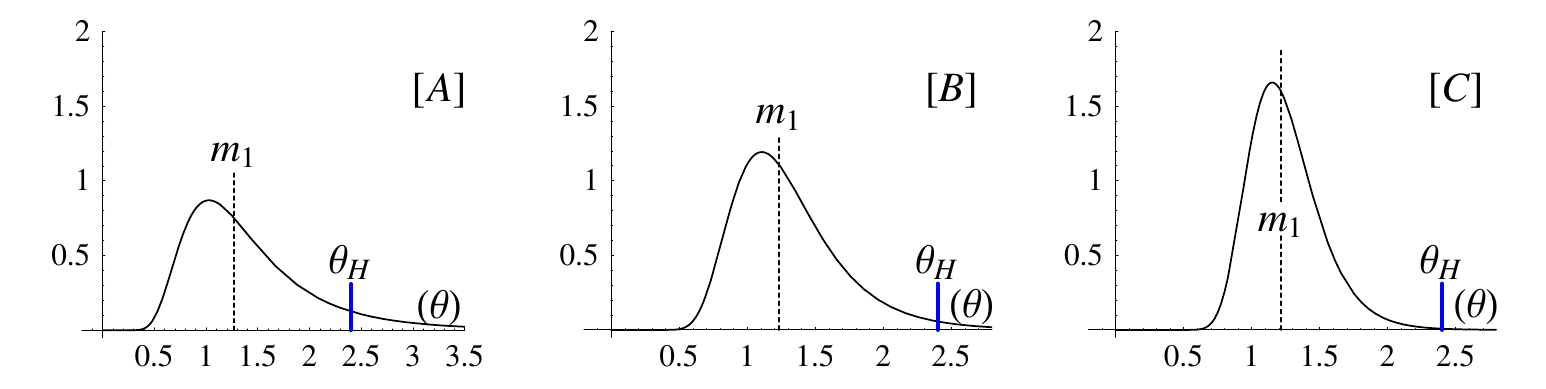}
    \caption{{\small Posterior density function  $g_1(\theta \vert  n\bar{x})$ and
 intervals  $I_H = (m_1,\theta_H)$ and $I_E = (\theta_H, \infty)$, using data from Example \ref{primoESEMPIO}. }}
 \label{fig2}
\end{figure}

Note that in all scenarios considered, we find the following relation between  $\delta_H$  and the \textit{p-value},
\begin{equation}
 \label{pvalue}
 \mbox{\textit{p-value}}=1 - \delta_H
 \end{equation}
 (in $[A]$ $\delta_H = 0.832$ and \textit{p-value}$= 0.168$, in $[B]$ $\delta_H = 0.96$ and  \textit{p-value}$= 0.04$, while in $[C]$ $\delta_H = 0.997$ and  \textit{p-value}$= 0.003$). This result  depends clearly on the use of the Jeffreys' prior, which is a matching prior for a scalar parameter (see \cite{ruli2021can}).
\begin{remark}
The fact that classical and Bayesian procedures, under certain conditions, produce the same conclusions is well known (see, for instance, \cite{lind65}). The linear relationship (\ref{pvalue}) also occurs in other simple cases. Even if it does not hold for more complicated models and in general for proper priors, it suggests a relationship between the traditional \textit{p-value} levels of significance $\{0.05, 0.01, 0.005, \dots \}$, and the critical values for the discrepancy measure $\{0.95, 0.99, 0.995, \dots\}$. In this paper we will not investigate the problem of the choice of the BDM threshold $\omega$. Several aspects about the choice of \textit{p-values} thresholds have been considered in \cite{benjamin2018} and can be suitably extended to the BDM. 
\end{remark}

Finally, to conclude Example \ref{primoESEMPIO}, it is useful to show the trend of the BDM when varying $n=1,2,\dots,25$ for six values of the MLE:  $(a)\ 0.8$, $(b)\ 1.2$, $(c)\ 1.6$ (case [A]) and $(d)\ 4.0$, $(e)\ 3.6$, $(f)\ 3.2$ (case [B]), see Figure \ref{fig:MLE}. In order to explain the difference between the BDM trends in cases [A] and [B], consider that:

\begin{enumerate}
    \item [(i)] in case [A] the posterior median $m_1 < \theta_H = 2.4$, whereas in case [B] $m_1 > \theta_H = 2.4$;
    \item [(ii)] $\delta_H$ is monotonically increasing, both with respect to $n$, and with respect to the distance $\vert m_1 - \theta_H \vert$;
    \item [(iii)] the posterior $g_1$ always has a positive asymmetry, which decreases as $n$ increases;
    \item [(iv)] the trend difference of the BDM in cases [A] and [B] depends on the fact that the posterior $g_1$ has `small' tails on the left-hand side of $m_1$ and `large' tails on the right-hand side.
\end{enumerate}

\begin{figure}
    \centering
    \includegraphics[scale=0.8]{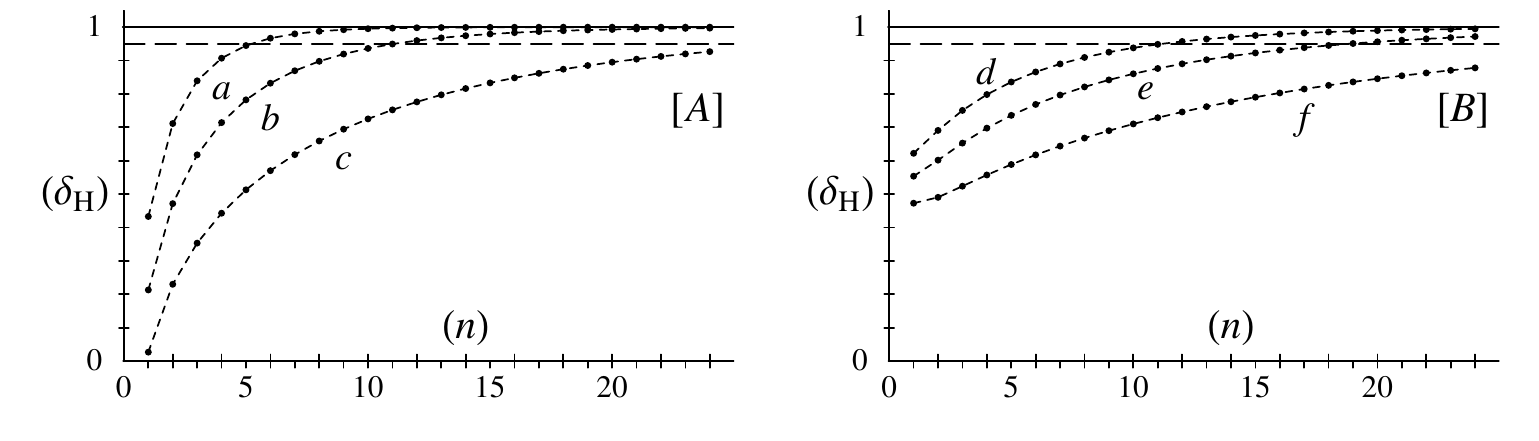}
    \caption{BDM for $n$ increasing and for different values of the MLE. Case $[A]$ with MLE $= 0.8\ (a),\ 1.2\ (b),\ 1.6\ (c)$ and case $[B]$ with MLE $= 3.2\ (f),\ 3.6\ (e),\ 4\ (d)$.}
    \label{fig:MLE}
\end{figure}
\end{example}

Moving forward in the discussion, in order to highlight the evaluative nature of the BDT, it is worth pointing out that it allows the separate and simultaneous testing of $\ell \geq 2$ hypotheses
\begin{equation}
 H_j: \, \varphi = \varphi_j,   \quad j=1,2,\dots , \ell\;,
 \end{equation}
%leading to a comparison between them, according to the values of the $\delta_{H_j}$, 
as shown in Example \ref{secondoESEMPIO}. Remember that with the comparative approach, among the $\ell$ competing hypotheses, only one is accepted. On the contrary, under the evaluative approach, it may happen that several hypotheses are supported by the data, or even that all hypotheses must be rejected.

\begin{example}\rm{- \textit{Evaluation of some hypotheses made by several experts (Bernoulli distribution)\\}}
\label{secondoESEMPIO}  
In the 1700s, several hypotheses $H_j : \theta = \theta_j$ were formulated about the birth masculinity rate
  $\theta =\frac{M}{M+F}$. Among them we consider  $\theta_1 = \frac{1}{2}$ (J. Bernoulli),
   $\theta_2 = \frac{13}{25}$ (J. Arbuthnot),
   $\theta_3 = \frac{1050}{2050}$ (J. P. S\"{u}ssmilch),
   $\theta_4 = \frac{23}{45}$ (P. S. Laplace).
 We assume that the gender of each newborn  is modeled as a $Bin(\cdot \vert 1, \theta)$.   Then, using data recorded in  1710 in London  (see, for instance, \cite{spiegelh2019}), with $7640$ males and  $7288$ females  (the MLE is $\hat{\theta} = 0.512$)  and  assuming the Jeffreys' prior $Beta(\theta \vert 1, 1)$, we compute  $\delta_{H_j}$  using the Normal  asymptotic approximation
\begin{equation*}
    \delta_{H_j} \cong1 - 2\cdot \displaystyle \int_{I_E^j} \Tilde{g}_1\big(\theta \vert \hat{\theta}, \textstyle{\frac{1}{n}}\hat{\theta}(1-\hat{\theta}) \big) {\rm d}\theta, \quad  j=1,2,3,4,
\end{equation*}
 with $\Tilde{g}_1$ the Normal distribution. Since $\delta_{H_1}=0.996$, $\delta_{H_2}=0.955$, $\delta_{H_3}=0.079$, $\delta_{H_4}=0.132$, we can conclude that the first two hypotheses has to be rejected, while there is not enough evidence to reject the hypotheses made by S\"{u}ssmilch and Laplace. 
\end{example}

\subsection{Examples of the more general case}
The examples presented hereafter, can be distinguished by tests concerning a parameter or a parametric function of a single population, and tests concerning  the comparison of two independent population parameters. 

\subsubsection{Tests  involving a single population}

\begin{example}\rm{- \textit{Test on the shape parameter, mean and variance of the Gamma distribution\\}}
\label{modelloGAMMA}
Let ${\boldsymbol{x}} = (x_1,\dots,x_n)$ be an \textit{iid}  sample of size $n$  from   $X \sim Gamma\big( x \vert \alpha , \beta \big)$,  $(\alpha , \beta) \in {\Real}^+ \times {\Real}^+$.  We denote by $m_g$  the geometric mean of $\boldsymbol{x}$.  The  likelihood function for $(\alpha,\beta)$ is given by
\begin{equation*}L(\alpha , \beta \vert \boldsymbol{x}) \; \propto \;
 \left(\displaystyle {\frac{\beta^\alpha}{\Gamma(\alpha)}} \cdot m_g^\alpha \cdot e^{- \bar{x}\cdot \beta} \right)^n \,.
\end{equation*}
For the fictitious data ${\boldsymbol{x}} = (0.8, 1.1, 1.2, 1.4, 1.8, 2, 4, 5, 8)$, we find that the  MLEs are $\hat{\alpha} = 1.921$  and $\hat{\beta} = 0.7572$. 
 
We are interested in testing the   hypotheses
$[A]$  $H_A: \alpha = \alpha_H$, with $\alpha_H = 2.5$,
$[B]$  $H_B: \mu = \mu_H$, with $\mu_H = 6$,  and
$[C]$  $H_C: \sigma^2 = \sigma^2_H$, with $\sigma^2_H = 2$, where $\mu = \displaystyle{\frac{\alpha}{\beta}}$ and $\sigma^2 = \displaystyle{\frac{\alpha}{\beta^2}}$ denote the mean and the variance of $X$. 

Adopting
the  Jeffreys' prior for $(\alpha , \beta)$,  i.e.
\begin{equation*}
g_0 ( \alpha , \beta) \; = \;
 g_0^\alpha ( \alpha ) \cdot g_0^\beta ( \beta ) \, \propto \;
 \displaystyle \sqrt{\alpha\cdot \psi^{(1)} (\alpha) - 1} \cdot {\frac{1}{\beta}} \, ,
\end{equation*}
 where $\psi^{(1)} (\alpha) = \sum_{j=0}^\infty (\alpha+ j)^{-2}$ denotes the \textit{digamma} function,
the posterior  for $(\alpha, \beta)$ is given by $g_1(\alpha , \beta \mid {\boldsymbol{x}}) \; = \; k\cdot g_0^\alpha ( \alpha ) \cdot g_0^\beta ( \beta ) \cdot L(\alpha , \beta \vert \boldsymbol{x}),$ with normalizing constant $k$. 
   \vspace{2 pt}

\noindent
 %$\blacktriangleright$
 \begin{itemize}
 \small
 \item  Case  $[A]$ \\
The hypothesis $H_A$ identifies the vertical straight line of equation $\alpha = \alpha_H$ and two subsets
${\boldsymbol{\Theta}}_a = \{ (\alpha, \beta) : \alpha < \alpha_H \}$
and
${\boldsymbol{\Theta}}_b = \{ (\alpha, \beta) : \alpha > \alpha_H \}$ (see Figure \ref{figuraGAMMA} [A]).
 Then we can compute
 \begin{eqnarray*}
 {\mathbb{P}} \big(  (\alpha, \beta) \in {\boldsymbol{\Theta}}_b\ \vert\ \boldsymbol{x}\big)&=&  \displaystyle
     \int_{\alpha_H}^\infty \int_0^\infty
   g_1(\alpha , \beta \mid {\boldsymbol{x}}) \,  {\rm d} \beta   \, {\rm d}\alpha  \\
   &=& \displaystyle
   k \cdot   \int_{\alpha_H}^\infty \int_0^\infty
   \sqrt{\alpha\cdot \psi^{(1)} (\alpha) - 1} \cdot {\frac{1}{\beta}}
   \left(\displaystyle {\frac{\beta^\alpha}{\Gamma(\alpha)}} \cdot m_g^\alpha \cdot e^{- \bar{x}\cdot \beta} \right)^n
   \,  {\rm d} \beta   \, {\rm d}\alpha  \\
   &=& \displaystyle \displaystyle
   k \cdot  \int_{\alpha_H}^\infty
   \sqrt{\alpha\cdot \psi^{(1)} (\alpha) - 1} \cdot  {\frac{\Gamma(n\alpha)}{\Gamma(\alpha)^n}}  \cdot \left({\frac{m_g}{n\, \bar{x}}} \right)^{n\alpha}
   \, {\rm d} \alpha \; = \; 0.215 \, ,
 \end{eqnarray*}
and  $\delta_H = 0.570$, a value that does not allows for the rejection of $H_A$.\bigskip

\item Case $[B]$ \\
The hypothesis $H_B$ identifies the straight line of equation $\beta = {\frac{1}{\mu_{H}}} \alpha$ in the $\alpha\beta$-plane (see Figure \ref{figuraGAMMA}  [\textit{B}])
and the two subsets

\begin{center}
${\boldsymbol{\Theta}}_c = \big\{ (\alpha, \beta) : \beta > {\frac{1}{\mu_{H}}}  \alpha  \big\}$
\, and \,
${\boldsymbol{\Theta}}_d= \big\{ (\alpha, \beta) : \beta < {\frac{1}{\mu_{H}}} \alpha  \big\}$.
\end{center}

We have

\begin{equation*}
  {\mathbb{P}} \big(  (\alpha,\beta ) \in {\boldsymbol{\Theta}}_d\ \vert \ \boldsymbol{x} \big)    \; = \;  \displaystyle
   \int_{{\boldsymbol{\Theta}}_d}
   g_1(\alpha , \beta \mid {\boldsymbol{x}})   \, {\rm d}\alpha\,  {\rm d} \beta
   \; = \;
    0.012     \,  ,
\end{equation*}
 and, since $\delta_H = 0.976$, we  reject $H_B$.\bigskip

 \item Case $[C]$ \\
The hypothesis $H_C$ identifies the parabola of equation $\beta = {\frac{1}{\sqrt{\sigma^2_{H}}}}  \sqrt{\alpha},$ in the $\alpha\beta$-plane  (see Figure \ref{figuraGAMMA}  [\textit{C}]), and the two subsets
 \begin{center}
 ${\boldsymbol{\Theta}}_e = \big\{ (\alpha,\beta) : \beta > {\frac{1}{\sqrt{\sigma^2_{H}}}} \sqrt{\alpha} \big\}$
\, and \,
 ${\boldsymbol{\Theta}}_f= \big\{ (\alpha, \beta) : \beta < {\frac{1}{\sqrt{\sigma^2_{H}}}} \sqrt{\alpha} \big\}$.
\end{center}

We have
\begin{equation*}
  {\mathbb{P}} \big(  (\alpha, \beta) \in {\boldsymbol{\Theta}}_e \ \vert \ \boldsymbol{x} \big)    \; = \;  \displaystyle
   \int_{{\boldsymbol{\Theta}}_e}
   g_1(\alpha , \beta \mid {\boldsymbol{x}})   \, {\rm d}\alpha\,  {\rm d} \beta
   \; = \;
    0.078     \,  .
\end{equation*}
Therefore $\delta_H = 0.846$, and so we do not reject $H_C$.

\end{itemize}

 \begin{figure}
    \centering
    \includegraphics[scale=0.8]{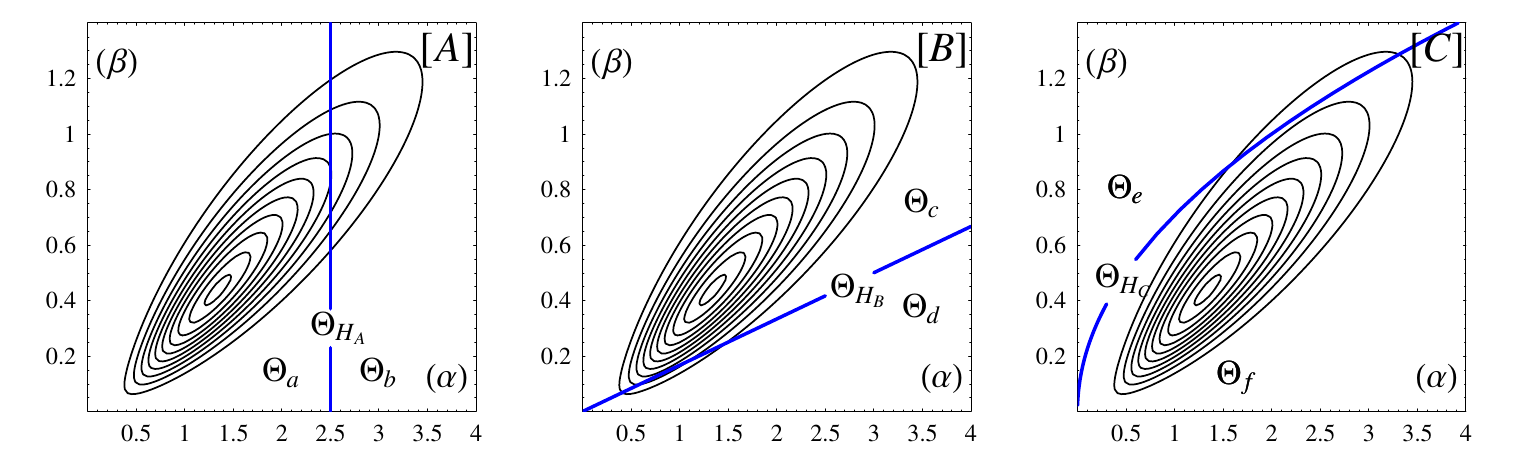}
    \caption{{\small Posterior density function $g_1(\alpha , \beta  \vert  {\boldsymbol{x}})$  from Example  \ref{modelloGAMMA} and corresponding sets of the induced partition in the cases $[A]$, $[B]$ and $[C]$.} 
    \label{figuraGAMMA}}
\end{figure}
%ev=0.622,0.711,0.835
\end{example}

\begin{example}\rm{- \textit{Test on the  coefficient of variation for a Normal distribution\\}}
\label{cv}
Given an  \textit{iid} sample  ${\boldsymbol{x}} = (x_1,\dots,x_n)$  from   $X \sim N \big( x \vert \mu , \phi^{-1} \big)$, the parameter of interest is $\psi = \displaystyle{\frac{\sqrt{Var(X)}}{\mid \E(X) \mid}} =\displaystyle{\frac{1}{\mid \mu \mid \sqrt{\phi}} }$. We are interested in testing  the hypothesis
\begin{equation*}
H: \psi = \psi_H,
\end{equation*}
with $\psi_H = 0.1$.
If we consider the Jeffreys' prior
$g_0(\mu,\phi) \propto  \phi^{-1} \cdot {\mathbf{1}}_{\Real \times \Real^+},$ the posterior  distribution is the Normal-Gamma density
 $$
 (\mu,\phi) \mid \boldsymbol{x}
  \; \sim  \;
  NG \big( \mu , \phi \mid \eta, \nu , \alpha, \beta \big) \, ,\,
 $$
with hyperparameters
 $( \eta, \nu, \alpha , \beta)$, where  $\eta = \bar{x}$,  $\nu = n$,  $\alpha = {\frac{1}{2}}(n-1)$, $\beta = {\frac{1}{2}}n s^2$,
 and density
$$g_1( \mu, \phi \mid \eta, \nu, \alpha , \beta)=\frac{\beta^{\alpha}\sqrt{\nu} }{\Gamma(\alpha)\sqrt{2 \pi}}  \phi^{\alpha-1/2} e^{-\frac{\nu \phi}{2} (\mu- \eta)^2}e^{-\beta\phi}.$$
We consider  the particular  case in which $\bar{x} = 17$ and  $s^2 = 1.6$ (so that  the MLE is $\hat{\phi} = 0.074$) with two samples of  size  $n= 10$ (Figure \ref{figura-coeffVAR} $[A]$) and $n= 40$ (Figure \ref{figura-coeffVAR} $[B]$).
In the $\mu\phi-$space, the hypothesis $H$ is  represented by the curve $\phi = \displaystyle{\frac{1}{\psi_H^2 }\mu^{-2}}$ and  determines the subsets  ${\boldsymbol{\Theta}}_a$ and  ${\boldsymbol{\Theta}}_b$ visualized in  Figure \ref{figura-coeffVAR}.

\begin{figure}
    \centering
    \includegraphics[scale=0.8]{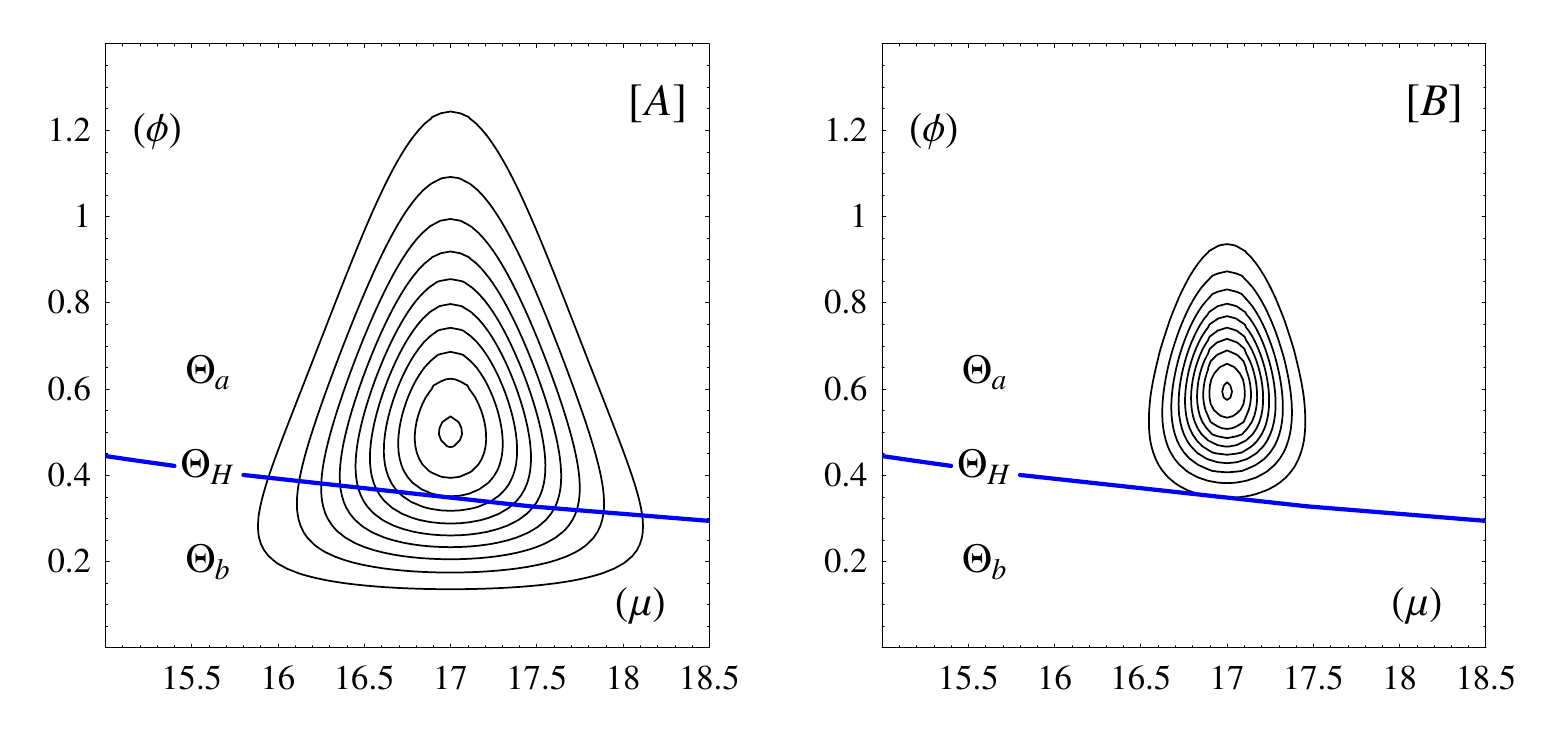}
    \caption{{\small Test on the coefficient of  variation  $\psi$ of a Gaussian population.  Data refers to Example \ref{cv}. In the plots, the sets  ${\boldsymbol{\Theta}}_a$, ${\boldsymbol{\Theta}}_b $  and  ${\boldsymbol{\Theta}}_H$  are reported for $n=10$ ($[A]$) and $n=40$ ($[B]$).}}
    \label{figura-coeffVAR}
\end{figure}

In case $[A]$ we have
\begin{equation*}
  {\mathbb{P}} \big(  (\mu, \phi) \in {\boldsymbol{\Theta}}_b\ \vert \ \boldsymbol{x} \big)=
   \int_{{\boldsymbol{\Theta}}_b}
  % NG \big( \mu , \phi   \mid  \eta_1, \nu_1 , \alpha_1 , \beta_1 \big)   \, {\rm d}\mu \,  {\rm d} \phi      = 0.208
   g_1( \mu, \phi \mid \eta, \nu, \alpha , \beta) \, {\rm d}\mu \,  {\rm d} \phi      = 0.215,
 \end{equation*}
 where $g_1( \mu, \phi \mid \eta, \nu, \alpha , \beta)$ is the Normal-Gamma density,
so that   $\delta_H = 0.570$ and we do not reject $H$. In case $[B]$,  we have  ${\mathbb{P}} \big(  (\mu, \phi) \in {\boldsymbol{\Theta}}_b \ \vert \ \boldsymbol{x} \big) = 0.014$ and,  since $\delta_H = 0.972$, we reject $H$. Therefore in such a case, with different sample sizes, the inferential conclusions change. 
% Bayesian e-value against H_0: 0.7197916  
% Bayesian e-value against H_0: 0.9813181
\end{example}

\begin{example}\rm{- \textit{Test on the skewness coefficient of the Inverse Gaussian distribution\\}}
\label{Mara}
Let us consider a Inverse Gaussian random variable $X$ with density 
\begin{equation*}
f(x \mid \mu, \nu)=\sqrt{\frac{\nu}{2\pi x^3}} \exp{\Big\{ -\frac{\nu}{2} \Big(\frac{x-\mu}{\mu \sqrt{x}}\Big)^2\Big\}} \cdot \textbf{1}_{\Real^+}(x),
\end{equation*}
where $(\mu,\nu) \in \Real^+ \times \Real^+$. The parameter of interest is the skewness coefficient $\gamma = 3 \sqrt{\frac{\mu}{\nu}}$ and it is of interest to test the hypothesis $H:  \gamma=\gamma_H$, where $\gamma_H = 2$. The Jeffreys' prior is 
\begin{equation*}
g_0(\mu,\nu) \propto  \frac{1}{\sqrt{\mu^3 \nu}} \cdot \textbf{1}_{\Real^+\times\Real^+}(\mu,\nu). \end{equation*}
Given $n$ observations,
the posterior distribution of $(\mu, \nu)$ is
\begin{equation*}
g_1(\mu, \nu \vert \boldsymbol{x})     \; \propto \;
\sqrt{\frac{\nu^{n-1}}{\mu^3}} \cdot
\exp{
\left\{-   \frac{n \ \nu}{2}\cdot 
 \left( \frac{\bar{x}}{\mu^2} - \frac{2}{\mu} + \frac{1}{a} \right)  \right\}
 }
 \cdot \textbf{1}_{\mathbb{R}^+\times\mathbb{R}^+}(\mu,\nu),
\end{equation*}
\noindent
where $\bar{x}$  and $a$ are the arithmetic and harmonic mean, respectively.

We apply the procedure to the following precipitation data (inches) from Jug Bridge, Maryland, analyzed in \cite{Folks:1978} (p. 272): 

\begin{center}
\begin{tabular}{ccccccc}
1.01 & 1.11 & 1.13 & 1.15 & 1.16\\ 
1.17 & 1.17 & 1.20 & 1.52 & 1.54\\ 
1.54 & 1.57 & 1.64 & 1.73 & 1.79\\ 
2.09 & 2.09 & 2.57 & 2.75 & 2.93\\
3.19 & 3.54 & 3.57 & 5.11 & 5.62.
\end{tabular}
\end{center}
The hypothesis
identifies in the parameter space 
$\boldsymbol{\Theta} =  \mathbb{R}^+ \times \mathbb{R}^+$ the
subsets
 \begin{displaymath}
 \begin{array}{ll}
 \boldsymbol{\Theta}_a 
 \;=&
 \Big\{(\mu,\nu)\in \boldsymbol{\Theta} :
3\sqrt{\frac{\mu}{\nu}}   <   \gamma_H \Big\}         \, ,        \\ 
 \boldsymbol{\Theta}_H
 \;=&
 \Big\{(\mu,\nu)\in  \boldsymbol{\Theta}:
  3 \sqrt{\frac{\mu}{\nu}}   =   \gamma_H \Big\}         \, ,        \\ 
 \boldsymbol{\Theta}_b
 \;=&
  \Big\{(\mu,\nu)\in  \boldsymbol{\Theta}:

3\sqrt{\frac{\mu}{\nu}}  >   \gamma_H \Big\}     \, .
  \end{array}
 \end{displaymath}
We have that
\begin{equation}
    \mathbb{P}\big((\mu,\nu) \in \boldsymbol{\Theta}_b
    \ \vert \     \boldsymbol{x} \big)  \; = \;
    \int_{\boldsymbol{\Theta}_b}    
    g_1(\mu, \nu \vert \boldsymbol{x})  \ {\rm d}\mu \; {\rm d}\nu = 0.078\, ,
\end{equation}
see Figure \ref{Mara}, then we obtain $\delta_H = 0.844$. This result indicates that we do not have enough evidence to reject the hypothesis $H$.

\begin{figure}
    \centering
    \includegraphics[width=0.4\textwidth]{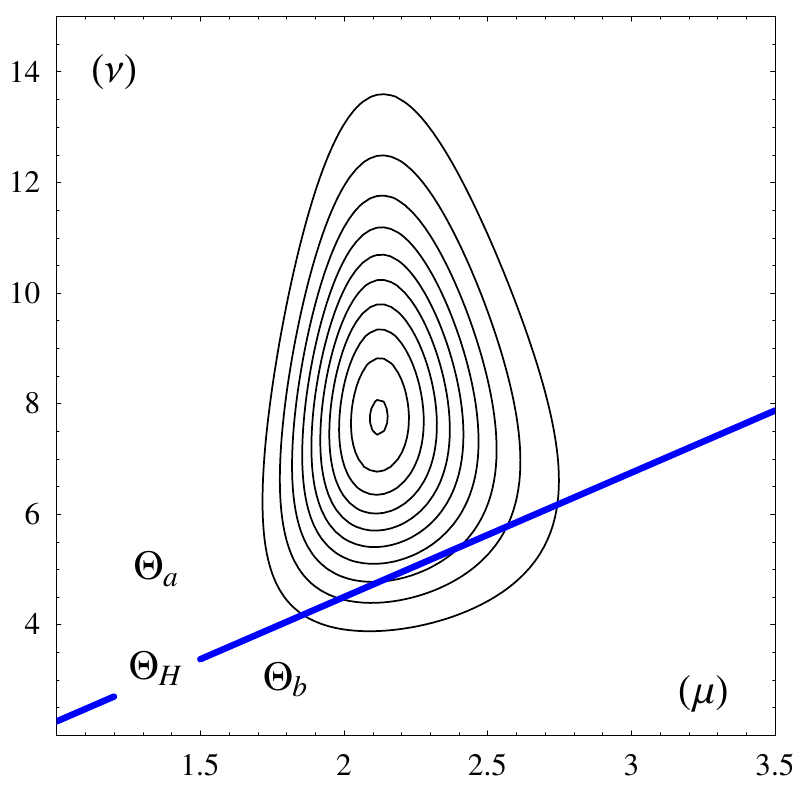}
    \caption{\small{Test on the skewness of the Inverse Gaussian distribution with $\gamma_H = 2$. In the plot  the sets of the partition induced by $H$ are reported. Data refers to Example \ref{Mara}.}}
 \label{Grafico}
\end{figure} 
\end{example}

\subsubsection{Tests involving two independent populations}
In this section we consider some examples concerning comparisons between parameters of two independent populations.

\begin{example}\rm{- \textit{Comparison between means  and precisions  of two independent Normal populations\\}}
\label{bisanzioESEMPIO}

Let us consider a case study on the dating of the core and periphery of some wooden furniture, found in a Byzantine church, using radiocarbon (see \cite{casber01}, p. 409). The  historians wanted to verify  if the mean age of the core is the same as the mean age of the periphery, using  two samples  of sizes  $m =14$  and $n = 9$, respectively, given by

\vspace{0.2cm}
\begin{center}
{\small
\begin{tabular}{c  rrrrr   c   rrr}
\textit{core}    &1294 &1279 &1274 &1264 &1263 \hspace{.4cm} &\textit{periphery}   &1284 &1272 &1256 \\
{}               &1254 &1251 &1251 &1248 &1240  \hspace{.4cm}      &{}             &1254 &1242 &1274 \\
{}               &1232 &1220 &1218 &1210 &{}    \hspace{.4cm}      &{}             &1264 &1256 &1250
\end{tabular} }
\end{center}
\vspace{0.2cm}
We assume that  the age of the core $X$ and of the  periphery $Y$ are distributed as
 $$
 X \sim N(x \vert \mu_1,\phi_1^{-1})  \quad {\rm and}  \quad  Y \sim N(y\vert \mu_2,\phi_2^{-1}),
 $$

 \noindent
 where  $Var(X) = \phi_1^{-1}$ and  $Var(Y) = \phi_2^{-1}$, and we assume that the data are \textit{iid} conditional on the parameters.  We consider for $(\mu_i,\phi_i)$ the Jeffreys' prior
 $$
  g_0^i(\mu_i,\phi_i) \propto  \phi_i^{-1} \cdot {\mathbf{1}}_{\Real \times \Real^+}
  \; , \; i=1,2 \, .
$$
We obtain
 $\bar{x} = 1249.86$,  $\bar{y} = 1261.33$,
 $\bar{d} = \bar{x} - \bar{y} = -11.48$,  while the MLEs for the sample standard deviations are $s_1 = 23.43$ and  $s_2 = 12.51.$
The posterior  distribution for $(\mu_i,\phi_i)$ is the Normal-Gamma law
 $$
 (\mu_i,\phi_i) \mid \boldsymbol{x}, \boldsymbol{y}
  \; \sim  \;
  NG \big( \mu_i , \phi_i \mid \eta_i, \nu_i , \alpha_i, \beta_i \big) \, , \;
 i=1,2,  \,
 $$
with hyperparameters
 $\eta_1 = \bar{x}$,  $\nu_1 = m$,  $\alpha_1 = {\frac{1}{2}}(m-1)$, $\beta_1 = {\frac{1}{2}}m s_1^2$
 \, and  \,
 $\eta_2 = \bar{y}$,  $\nu_2 = n$,  $\alpha_2 = {\frac{1}{2}}(n-1)$, $\beta_2 = {\frac{1}{2}}n s_2^2$,
 and density
$$g^i_1( \mu_i, \phi_i \mid \eta_i, \nu_i, \alpha_i , \beta_i)=\frac{\beta_i^{\alpha_i}\sqrt{\nu_i} }{\Gamma(\alpha_i)\sqrt{2 \pi}}  \phi_i^{\alpha_i-1/2} e^{-\frac{\nu_i \phi_i}{2} (\mu_i- \eta_i)^2}e^{-\beta_i\phi_i}, \; i=1,2.$$
The hypothesis
of interest  $$H_A : \mu_1 -\mu_2 = 0,  \quad \forall \phi_1 > 0, \quad \forall \phi_2 > 0,$$
 identifies the following  subsets in the parameter space
 \begin{equation*}
\begin{array}{ll}
    {\boldsymbol{\Theta}}_a &= \ \Big\{ \,  {\Real}^2 \times {\Real}_+^2 :  \mu_1 < \mu_2 \; \Big\},\\
    {\boldsymbol{\Theta}}_{H_A} &= \ \Big\{ \,  {\Real}^2 \times {\Real}_+^2 :  \mu_1 = \mu_2 \; \Big\},\\
    {\boldsymbol{\Theta}}_b &= \ \Big\{ \,  {\Real}^2 \times {\Real}_+^2 :  \mu_1 > \mu_2 \; \Big\}.
\end{array}     
 \end{equation*}
Then  we can  compute
 $$
  {\mathbb{P}}\big(  (\mu_1, \mu_2, \phi_1, \phi_2) \in {\boldsymbol{\Theta}}_a \ \vert \ \boldsymbol{x}, \boldsymbol{y}\big) \hspace{8cm}
 $$
 \begin{displaymath}
   \begin{array}{ll} \vspace{4 pt}
   \hspace{0.5 cm}
   &=\;   \displaystyle
   \int_{{\boldsymbol{\Theta}}_a} \prod_{i=1}^{2}
   g^i_1( \mu_i, \phi_i \mid \eta_i, \nu_i, \alpha_i, \beta_i)
   \, {\rm d}\mu_1 \, {\rm d}\mu_2  \, {\rm d} \phi_1 \, {\rm d} \phi_2      \\  \vspace{3 pt}
   &=\;   \displaystyle
   \int_{\mu_1< \mu_2} \prod_{i=1}^{2} \frac{\Gamma(\alpha_i+\frac{1}{2})}{\Gamma(\alpha_i)}\Big(\frac{\nu_i}{2 \pi \beta_i}\Big)^{1/2} \Big[1+\frac{\nu_i}{2 \beta_i}(\mu_i-\eta_i)^2\Big]^{-(\alpha_i+\frac{1}{2})}
   \, {\rm d}\mu_1 \, {\rm d}\mu_2          \\
   &=\; \displaystyle
   0.089 \, ,
   \end{array}
 \end{displaymath}
so we have $\delta_H = 0.823$, a value that do not lead to the rejection of the hypothesis. 
 \noindent
We exploited the fact that  the marginal of each   $\mu_i$ is a
 Generalized Student's t-distribution (denoted by \textit{StudentG}) with hyperparameters $\big( \eta_i, \frac{\beta_i}{ \nu_i\alpha_i}, 2 \alpha_i \big)$.
\normalsize
\begin{figure}
    \centering
    \includegraphics[scale=0.8]{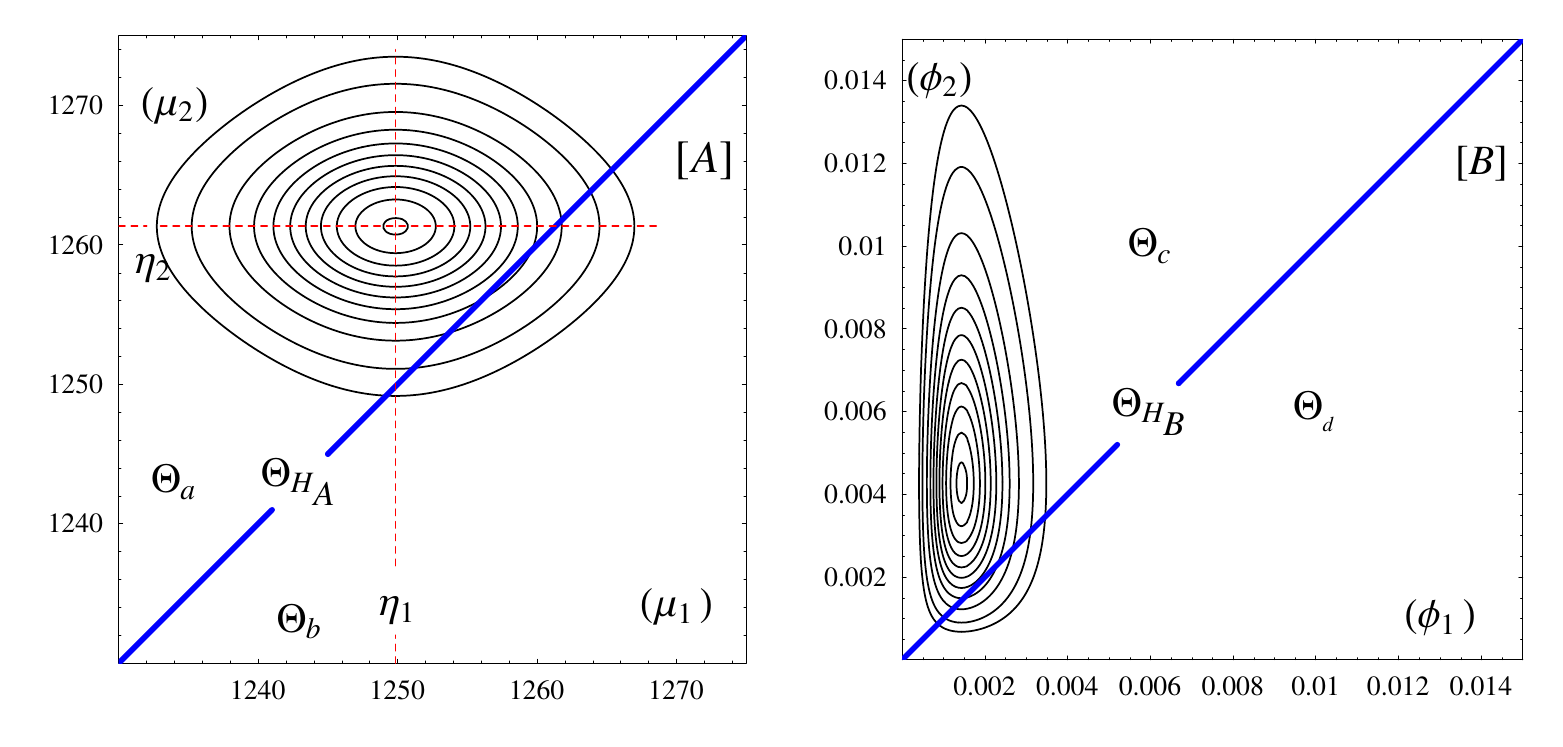}
    \caption{\small{Comparisons between means ([A]) and precisions ([B]) of independent normal populations for data in Example  \ref{bisanzioESEMPIO}. For both cases we show the contour plots of the marginals of $\mu_j$ ([A]) and $\phi_j$ ([B]), and the partition sets associated with the corresponding hypotheses.}}  
    \label{figura-BF}
\end{figure}
Figure  \text{\ref{figura-BF}} $[A]$
in the space  $(\mu_1, \mu_2)$ shows the contour lines of the distribution  
$$StudentG \Big( \mu_1   \big\vert \, \eta_1, {\frac{\beta_1}{\nu_1 \cdot \alpha_1}}, 2\alpha_1 \Big) \cdot StudentG \Big( \mu_2   \big\vert \, \eta_2, {\frac{\beta_2}{\nu_2 \cdot \alpha_2}}, 2\alpha_2 \Big).$$ Note that the  homoscedasticity assumption is not necessary.
Consider now the hypothesis
$$H_B : \phi_1 - \phi_2 =0,  \quad  \forall \mu_1, \mu_2,$$
which determines in the parameter space the subsets
\begin{equation*}
    \begin{array}{ll}
        {\boldsymbol{\Theta}}_c  &= \; \Big\{ \,  {\Real}^2 \times {\Real}_+^2 :  \phi_1 < \phi_2 \; \Big\},  \\
        {\boldsymbol{\Theta}}_{H_B}  &= \; \Big\{ \,  {\Real}^2 \times {\Real}_+^2 :  \phi_1 = \phi_2 \; \Big\}, \\
        {\boldsymbol{\Theta}}_d  &= \; \Big\{ \,  {\Real}^2 \times {\Real}_+^2 :  \phi_1 > \phi_2 \; \Big\}.
    \end{array}
\end{equation*}
We have
\small
 \begin{equation*}
 {\mathbb{P}} \big(  (\mu_1, \mu_2, \phi_1, \phi_2) \in {\boldsymbol{\Theta}}_c \ \vert \ \boldsymbol{x}, \boldsymbol{y} \big)
   =
   % \int_{{\boldsymbol{\theta}}_c}   NG \big( \mu_1 , \phi_1   \mid  \eta_1, \nu_1 , \alpha_1 , \beta_1 \big) \cdot
 %  NG \big( \mu_2 , \phi_2   \mid  \eta_2, \nu_2 , \alpha_2 , \beta_2 \big)
  % \, {\rm d}\mu_1 \, {\rm d}\mu_2  \, {\rm d} \phi_1 \, {\rm d} \phi_2     =   \\
   \int_{\phi_1< \phi_2} \prod_{i=1}^{2} \frac{\beta_i^{\alpha_i}}{\Gamma(\alpha_i)}\phi_i^{\alpha_i-1}e^{-\phi_i\beta_i} \, {\rm d}\phi_1 \, {\rm d}\phi_2  = 0.046,
      \end{equation*}
\normalsize
 \noindent
from which it follows that  $\delta_H =  0.908$ and we reject the hypothesis $H$. To compute the integral we have used the fact that  the marginal of each  $\phi_i$ has Gamma distribution 
 with parameters $(\alpha_i, \beta_i), \ i=1,2$.  
 
 The contour lines of the law  $Gamma ( \phi_1  \vert  \alpha_1,  \beta_1 ) \cdot Gamma ( \phi_2  \vert  \alpha_2,  \beta_2),$  in  the space  $(\phi_1, \phi_2)$, are reported in Figure \text{\ref{figura-BF}} $[B]$. 
% \, $\lhd$
\end{example}

\begin{example}\rm{- \textit{Comparison of the  shape parameter of two Gamma distributions}}
\label{2Gamma}
 \rm{\,} \\
Let us consider two  \textit{iid}  Gamma populations $X_i  \sim Gamma \big( \alpha_i , \beta_i \big), $  $(\alpha_i , \beta_i) \in {\Real}^+ \times {\Real}^+$ , $ i=1,2$, and let us consider two samples of sizes $n_1=9$ and $n_2=12$, respectively,  with  sample means $\bar{x}_1 = 2.811$  and  $\bar{x}_2 = 1.973$,  and  geometric means   $m_{g_1} = 2.116$ and $m_{g_2} = 1.327$.

We are interested in testing  $H: \alpha_1 = \alpha_2.$ The posterior  distribution for $(\alpha_1, \beta_1, \alpha_2, \beta_2)$ is given by

\begin{figure}
    \centering
    \includegraphics[scale=0.8]{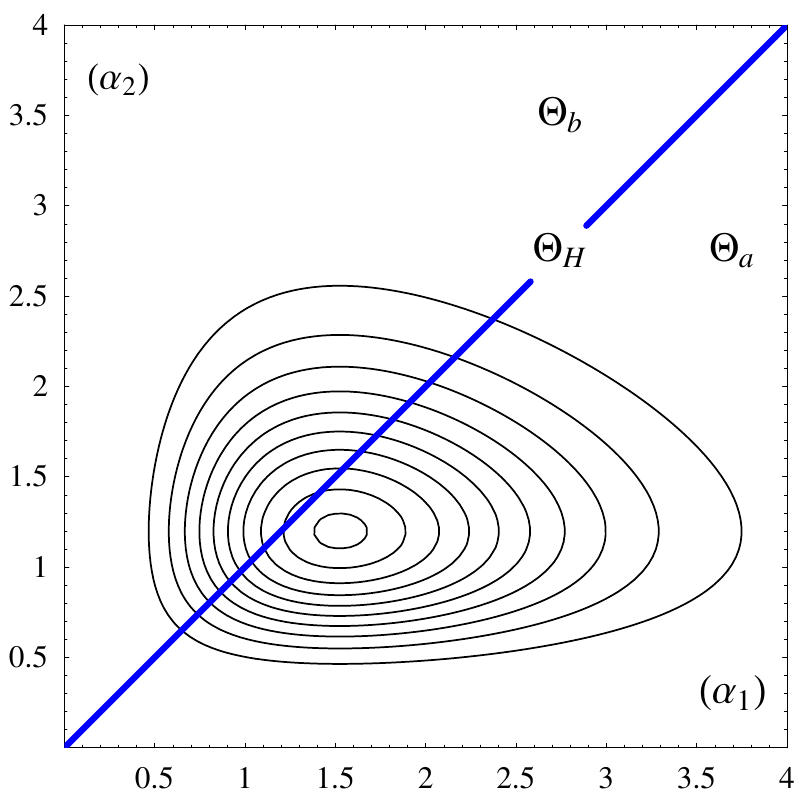}
    \caption{\small{Comparison of the shape parameters of two independent Gamma populations,  using data of  Example \ref{2Gamma}.  The sets   $\Theta_a,  \Theta_b$ and  $\Theta_H$ of the partition are reported. } } 
    \label{confrA}
\end{figure}

 $$
g_1( \alpha_1 , \beta_1, \alpha_2 , \beta_2 \vert  {\boldsymbol{x}_1} , {\boldsymbol{x}_2})=
 g_1^1 (\alpha_1 , \beta_1 \vert {\boldsymbol{x}_1}) \cdot  g_1^2 (\alpha_2, \beta_2 \vert {\boldsymbol{x}_2} ) \, ,
 $$
where
$$g_1^i(\alpha_i, \beta_i \vert {\boldsymbol{x}_i} ) = k_i \cdot g_0 ( \alpha_i,  \beta_i ) \cdot L(\alpha_i , \beta_i \mid {\boldsymbol{x}_i}),$$  with normalizing constant $k_i$, $i=1,2$.
Let  $\boldsymbol{\Theta}_a = \big\{(\alpha_1,\alpha_2) \in \Real^+ \times \Real^+ : \alpha_1 > \alpha_2\big\}$
 and
 $\boldsymbol{\Theta}_b = \big\{(\alpha_1,\alpha_2) \in \Real^+ \times \Real^+ : \alpha_1 < \alpha_2\big\}$ 
  (see Figure {\ref{confrA}}).  In order to test the hypothesis $H$, we compute the probability
 \begin{align*}
     &{\mathbb{P}} ( (\alpha_1,\alpha_2) \in \boldsymbol{\Theta}_b \ \vert \ {\boldsymbol{x}_1}, {\boldsymbol{x}_2} )\\
     &=\displaystyle \int_{\alpha_1 < \alpha_2} \int_{{\Real}^+ \times {\Real}^+} g_1^1 (\alpha_1 , \beta_1 \vert {\boldsymbol{x}_1} ) \cdot g_1^2(\alpha_2 ,\beta_2 \vert {\boldsymbol{x}_2} ) \,  {\rm d} \beta_1  {\rm d} \beta_2  \, {\rm d}\alpha_1  \, {\rm d}\alpha_2     \\ 
 &=   \displaystyle \int_{\alpha_1 < \alpha_2}   \prod_{i=1}^2 k_i  \cdot g_0^\alpha(\alpha_i)   \cdot {\frac{\Gamma(n_i \alpha_i)}{\Gamma(\alpha_i)^{n_i}}}  \cdot  \left(\frac{{m_{g_i}}}{n_j \bar{x}_i} \right)^{n_i \alpha_i} {\rm d}\alpha_1 {\rm d}\alpha_2=    0.311 
 \end{align*}
and, since $\delta_H = 0.378$, we do not reject $H$.
\end{example}
%------------------------------------------------------------------------------------------------

%------------------------------------  ------------------------------------------------
\section{Comparison with the FBST}\label{sec3}
 
In this section we present a comparison of the BDT with the Full Bayesian Significance Test (FBST) as presented  in \cite{perstern20}, which provides an overview of the $e$-value.

In order to facilitate the discussion, let us briefly review the definition of the $e$-value and the related testing procedure. The FBST can be used with any standard parametric statistical model, where $\boldsymbol{\theta} \in \Theta\ \subseteq\ \Real^p$. It tests a sharp hypothesis $H$ which identifies the null set $\Theta_{H}$. The conceptual approach of the FBST consists of determining the $e$-value that represents the Bayesian evidence against $H$. To construct this measure, the authors introduce the posterior surprise function and its supremum, given respectively by
\begin{equation*}
    s(\boldsymbol{\theta}) = \frac{g_1(\boldsymbol{\theta} \vert \boldsymbol{x})}{r(\boldsymbol{\theta})} \quad \text{ and }\quad s^*=s(\boldsymbol{\theta}^*) = \sup_{\boldsymbol{\theta} \in \Theta_{H}} s(\boldsymbol{\theta}),
\end{equation*}
where $r(\boldsymbol{\theta})$ is a suitable reference function to be chosen. Then, a tangential set is defined as
\begin{equation*}
    \overline{T}(s^*) = \{\boldsymbol{\theta} \in \Theta \vert s(\boldsymbol{\theta}) > s^*\},
\end{equation*}
to the sharp hypothesis $H$, also called a Highest Relative Surprise Set (HRSS), which includes all parameter values $\boldsymbol{\theta}$ that attain a larger surprise function value than the supremum $s^*$ of the null set. Finally, the $e$-value, that represents the Bayesian evidence against $H$, is defined as 
\begin{equation*}
    \overline{ev}(H) = \overline{W}(s^*) = \int_{\overline{T}(s^*)} g_1(\boldsymbol{\theta} \vert \boldsymbol{x}) \ d\boldsymbol{\theta}.
\end{equation*}
On the contrary,  the $e$-value in support of $H$ is $ev(H) = 1 - \overline{ev}(H_0)$, which is evaluated by means of the set $T(s^*)=\Theta  \setminus \overline{T}(s^*)$ and the cumulative surprise function $W(s^*) = 1 - \overline{W}(s^*)$. In conclusion, the FBST is the procedure that rejects $H$ whenever $\overline{ev}(H)$ is large.

As pointed out in \cite{perstern20} (Section 3.2) ``\textit{the role of the reference density is to make $\overline{ev}(H)$ explicitly invariant under suitable transformations of the coordinate system}''.  A first non-invariant definition of this measure, which corresponds to the use of a flat reference function $r(\theta)\propto 1$ in the second formulation, has been given in \cite{perstern99}. The first version involved the determination of the tangential set $\overline{T}$ starting only from the posterior distribution, whereas in the second, a corrective element has been introduced by also including the reference function. Some of the suggested choices for the reference function are the use of uninformative priors such as ``\textit{the uniform, maximum entropy densities, or Jeffreys’ invariant prior}'' (see \cite{perstern20}, Section 3.2). \\

\subsection{Similarities and differences between the procedures}
The most striking similarity between the FBST and the BDT is that both tests, fully accepting the likelihood principle and relying on the posterior distribution of the parameter $\boldsymbol{\theta} \in \boldsymbol{\Theta}$, are clearly Bayesian.

Another important similarity is that, asymptotically, both tests lead to the rejection of the hypothesis $H$ when it is false (i.e. when we test $\theta_H \neq \theta^*$ where $\theta^*$ is the true value of the parameter).  On the contrary, if $\theta_H=\theta^*$ they have a different asympototic behaviour (see Proposition \ref{prop} for the BDM and Section 3.4 in \cite{perstern20} for the $e$-value). 

Certainly, the FBST has a more general reach than the BDT. Indeed, it examines the entire class of sharp hypotheses, whereas the extension of the BDT to such hypotheses is not straightforward and, currently, is limited to considering the subclass of the hypotheses expressed as $H:\varphi=\varphi_H$ that are able to partition the parameter space $\boldsymbol{\Theta}$ as $\big\{ \boldsymbol{\Theta}_a, \, \boldsymbol{\Theta}_H, \, \boldsymbol{\Theta}_b \big\}$. Moreover, notice that while the integration sets $\boldsymbol{\Theta}_a$ and $\boldsymbol{\Theta}_b$ are determined exclusively by the hypothesis, the tangential set $\overline{T}$ depends on the hypothesis, the posterior density and the choice of the reference function. It is questionable, on the other hand, whether the $e$-value is as easily computable as the BDM is in cases where the parameter space has dimension higher than 1.

Unlike the BDM, the elimination of nuisance parameters is not recommended when using the $e$-value. In fact, this measure is not invariant with respect to marginalisations of the nuisance parameter and the use of marginal densities to construct credible sets may produce inconsistency.

It is easy to see that one can create an analogy between the \textit{p-value}, the $e$-value and $\delta_H$. Regarding frequentist \textit{p-value}s, the sample space is ordered according to increasing inconsistency with the assumed null hypothesis $H$. The FBST instead orders the parameter space according to increasing inconsistency with the assumed null hypothesis $H$, based on the concept of statistical surprise. In the same way, it can be seen that the probability in $\eqref{formula:tail1}$ has to do with the posterior probability of exceeding $\theta_H$ in a direction in contrast with the data (namely, the side where there is more posterior probability).

Another similarity occurs when considering the reference density $r(\theta)$ as the (possibly improper) uniform density, since the first and second definitions of evidence define the same tangent set, i.e. the HRSS and the HPDS coincide. Then, for a scalar parameter $\theta$, since the BDM is linked to the equi-tailed credible regions while the $e$-value is linked to the HPDS, we have that if:
\begin{itemize}
  \item  $g_1(\theta  \vert \boldsymbol{x})$ is symmetric and unimodal, then
  $\overline{ev}(H) = \delta_H$;

  \item $g_1(\theta  \vert \boldsymbol{x})$  is asymmetric and unimodal (for instance with positive skewness)  and $m_1 < \theta_H$  [$\theta_H < m_1$],  then
    $\overline{ev}(H) > \delta_H$ [$\overline{ev}(H) < \delta_H$]. When $m_1=\theta_H$  we have $0=\delta_H < \overline{ev}(H)$.
\end{itemize}

\subsubsection{Simulation study}
In order to determine the resulting false-positive rates of both the FBST and the BDT, we conduct a simulation study for specific sample sizes. 

Let  $\boldsymbol{x}=(x_1, \dots, x_n)$  be an \textit{iid} sample  of size $n$ from the exponential distribution $X \sim Exp\big(x \vert 1/\theta^* \big)$, with  $\theta^*=1.2$. We are interested in testing the hypothesis $H: \theta_H=\theta^*=1.2$.  Assuming a Jeffreys' prior $g_0(\theta) \propto \theta^{-1}$,  the posterior distribution is an $InvGamma(\theta \vert n, \sum x_i) $; see Example \ref{primoESEMPIO}.

Table \ref{tab:erroreI} shows the simulation results for three different values of the threshold $\omega=\{0.90,0.95,0.99\}$, for $S=50000$ simulations and $D=50000$ posterior draws. Across the different sample sizes considered, the false-positive rates are very similar for both tests and, as we expect since we are using objective priors (see \cite{bayarri2004interplay}), they are close to the error of the first type $\alpha=\{0.10, 0.05, 0.01\}$, related to $\omega$. Similar results, not reported here, were found adopting a Poisson model.

\begin{table}[!htbp]
\resizebox{\columnwidth}{!}{
{\setlength{\extrarowheight}{5pt}
\begin{tabular}{|l|llllllllllll|}
\hline
\multicolumn{1}{|c|}{} & \multicolumn{4}{c}{$\boldsymbol{\omega = 0.90}$}&\multicolumn{4}{c}{$\boldsymbol{\omega = 0.95}$}& \multicolumn{4}{c|}{$\boldsymbol{\omega = 0.99}$}\\ \hline
\multicolumn{1}{|c|}{} & \multicolumn{4}{c}{$n$} & \multicolumn{4}{c}{$n$} & \multicolumn{4}{c|}{$n$}\\
\multicolumn{1}{|c|}{} & \multicolumn{1}{c}{10} & \multicolumn{1}{c}{100} & \multicolumn{1}{c}{1000} & \multicolumn{1}{c|}{10000} & \multicolumn{1}{c}{10} & \multicolumn{1}{c}{100} & \multicolumn{1}{c}{1000} & \multicolumn{1}{c|}{10000} & \multicolumn{1}{c}{10} & \multicolumn{1}{c}{100} & \multicolumn{1}{c}{1000} & \multicolumn{1}{c|}{10000} \\ \cline{1-13} 
{\textit{\textbf{\begin{tabular}[c]{@{}l@{}}e-value\\ $r(\theta) \propto 1$\vspace{0.1cm}\end{tabular}}}}       
& 0.102  & 0.100 & 0.099 & \multicolumn{1}{l|}{0.098} & 0.052 & 0.050 & 0.051 & \multicolumn{1}{l|}{0.049} & 0.011 & 0.010 & 0.011 & 0.011 \\ \cline{1-1}
{\textit{\textbf{\begin{tabular}[c]{@{}l@{}}e-value\\ $r(\theta)=g_0(\theta)$\vspace{0.1cm}\end{tabular}}}}       & 0.101 &  0.102 & 0.100 & \multicolumn{1}{l|}{0.098} & 0.051 & 0.050 & 0.051 & \multicolumn{1}{l|}{0.049} & 0.010 & 0.010 & 0.011 & 0.011 \\ \cline{1-1}
\textit{$\boldsymbol{\delta_H}$}\vspace{0.25cm} & 0.103  & 0.102 & 0.101 & \multicolumn{1}{l|}{0.099} & 0.053  & 0.049 & 0.052 & \multicolumn{1}{l|}{0.049} & 0.010 & 0.009 & 0.011 & 0.011 \\ \hline
\end{tabular}}}
\vspace{0.1cm}
\caption{False positive rates for different sample sizes $n$ and different thresholds $\omega$.}
\label{tab:erroreI}
\end{table}

\subsubsection{Some Examples}
\vspace{0.3cm}
In order to compare the BDM and the $e$-value, let us consider different situations and then examine the results.

\begin{example}  \rm{\,}
 ({\textit{Continuation of Example \ref{primoESEMPIO}}}) \\
As a first comparative scenario, consider the test performed in Example \ref{primoESEMPIO} in which $\theta_H=2.4$ and additionally the case in which $\theta_H=0.7$. Since the posterior  $g_1(\theta\vert  \boldsymbol{x})$ has a positive skewness and $m_1 < \theta_H=2.4$ then $\overline{ev}(H) > \delta_H$, on the contrary, for $m_1 > \theta_H=0.7$ then $\overline{ev}(H) < \delta_H$. Indeed, we find the results reported in Table \ref{tab:Ex1}.
\begin{table}[!htbp]
\centering
\begin{tabular}{l|ccc|ccl|}
\multicolumn{1}{l|}{} & \multicolumn{3}{c|}{\textbf{$\theta_H=2.4$}}                                                                                        & \multicolumn{3}{c|}{\textbf{$\theta_H=0.7$}}                                                                                                             \\ \hline
                      & \multicolumn{2}{c|}{\textit{\textbf{e-value}}}                                      & \multirow{2}{*}{\textit{\textbf{$\boldsymbol{\delta_H}$}}} & \multicolumn{2}{c|}{\textit{\textbf{e-value}}}                                      & \multicolumn{1}{c|}{\multirow{2}{*}{\textit{\textbf{$\boldsymbol{\delta_H}$}}}} \\
\textit{\textbf{}}    & \textit{$r(\theta)\propto 1$} & \multicolumn{1}{c|}{\textit{$r(\theta)=g_0(\theta)$}} &                                               & \textit{$r(\theta)\propto 1$} & \multicolumn{1}{c|}{\textit{$r(\theta)=g_0(\theta)$}} & \multicolumn{1}{c|}{}                                              \\ \hline
{[}A{]} $  n = 6$              & 0.909                         & \multicolumn{1}{c|}{0.866}                          & 0.832                                         & 0.646                         & \multicolumn{1}{c|}{0.847}                          & 0.886                                                              \\
{[}B{]} $n=12$              & 0.978                         & \multicolumn{1}{c|}{0.968}                          & 0.960                                         & 0.899                         & \multicolumn{1}{c|}{0.957}                          & 0.968                                                                   \\
{[}C{]} $n=24$               & 0.999                         & \multicolumn{1}{c|}{0.998}                          & 0.997                                         & 0.991                         & \multicolumn{1}{c|}{0.997}                          & 0.997         \end{tabular}
\vspace{0.4cm}
\caption{The table shows, for the 3 different cases examined in Example \ref{primoESEMPIO}, the values of $\delta_H$ and of the $e$-value considering, as a reference distribution, both a flat reference function and a Jeffreys' prior.}
\label{tab:Ex1}
\end{table}
\end{example}
The differences between the $e$-value and $\delta_H$, which in this example appear to be modest, can actually become meaningful when the posterior has a greater asymmetry and heavy tailes. In such case, comparing different hypotheses, the FBST always leads to favour the hypothesis with higher density. Moreover, the $e$-value may be more or less robust w.r.t. the  position of $\theta_H$, as it is highlighted in the example below.
\begin{example}  \rm{- \textit{Test on the mean of the Inverse Gaussian distribution}}
\rm{\,}\\
\label{GI}
Consider a random variable $X$ with  Inverse Gaussian distribution $X \sim IG(x \vert \mu, \nu_0)$,   $\mu \in {\Real}^+$ and  $\nu_0$ known.  Given an \textit{iid}  sample $\boldsymbol{x}$ of size $n$, the likelihood  function   for $\mu$ is
  $L(\mu\vert\boldsymbol{x}) \propto
  \exp \left\{ -n \nu_0 \cdot\left( \frac{\bar{x}}{2\mu^2}- \frac{1}{\mu}\right) \right\}.$
Adopting  the Jeffreys' prior  $g_0(\mu) \propto \frac{1}{\sqrt{\mu^3}}$, we obtain the posterior   
\begin{equation*}g_1(\mu\vert  {\boldsymbol x}) \propto \frac{1}{\sqrt{\mu^3}}
 \cdot
 \exp \left\{ -n \nu_0 \cdot\left( \frac{\bar{x}}{2\mu^2}- \frac{1}{\mu}\right) \right\}.
\end{equation*}
 We are interested in testing the hypothesis $H: \, \mu = \mu_H$ and we consider a sample of size $n=8$ for which $\bar{x}  = 4.2$ and $m_1 = 4.483$. 
For $\nu_0=5$, we choose to test $H_A: \mu = 2.5$ and $H_B: \mu = 12$. The results of the analysis are displayed in Table \ref{tab:IG} and Figure \ref{figuraPROC}. If we choose $\omega = 0.95$ as a rejection threshold in both cases, and with both references, we are lead to opposite inferential conclusions. 

\vspace{0.4cm}
  
\begin{table}[!htbp]
\centering
\begin{tabular}{c|cc|c}
                   & \multicolumn{2}{c|}{\textit{\textbf{e-value}}}                   & \multirow{2}{*}{\textit{$\boldsymbol{\delta_H}$}} \\
\textit{\textbf{}} & \textit{$r(\theta)\propto 1$} & \textit{$r(\theta)=g_0(\theta)$} &                                               \\ \hline
$H_A: \mu = 2.5$   & 0.803                             & 0.848                                & 0.975                                         \\
$H_B: \mu = 12$    & 1                             & 1                                & 0.907             \end{tabular}
\vspace{0.4cm}
\caption{For the two different hypothesis examined in Example \ref{GI}, the table shows $\delta_H$ and the $e$-value considering, as a reference distribution, both a flat reference function and a Jeffreys' prior.}
\label{tab:IG}
\end{table}

\begin{figure}
    \centering
    \includegraphics[scale=0.8]{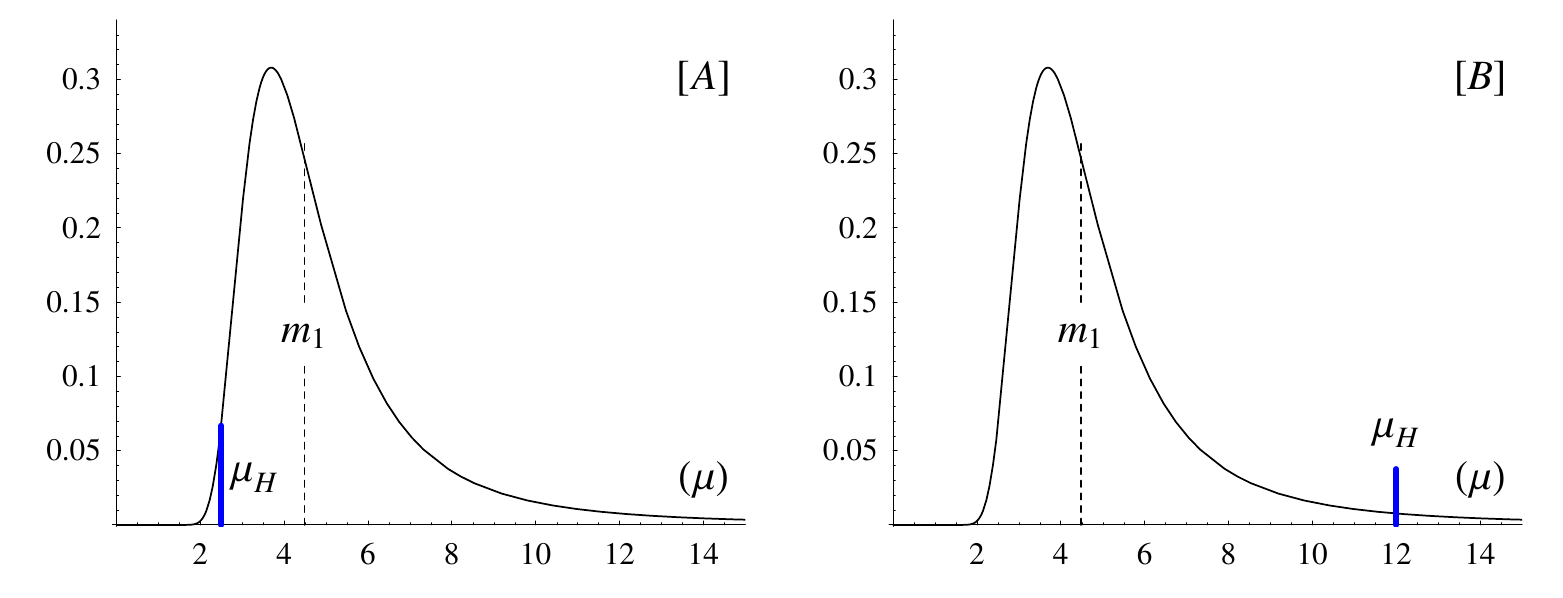}
    \caption{{\small Posterior density function  $g_1(\mu \vert \boldsymbol{x} )$ associated to  Example \ref{GI}. In   $[A]$ we have  $\mu_H = 2.5 < m_1$, while  in $[B]$ $\mu_H = 12 > m_1.$
    \label{figuraPROC}}
}
\end{figure}
\end{example}

\begin{example}\rm{\,}
 ({\textit{Continuation of Examples \ref{modelloGAMMA}, \ref{cv}, \ref{Mara}}}) \\
Let us now compare the results obtained with the FBST and the BDT for the Examples  \ref{modelloGAMMA}, \ref{cv} and \ref{Mara}, when fixing a value of 0.95 as a rejection threshold. 

The conclusions reached with the FBST and with the BDT for Example \ref{modelloGAMMA}, which can be seen in Table \ref{3fbst}, are the same (for both reference functions considered) although, in some cases, there are substantial differences between the values of the evidence measures. To summarise, the hypothesis $H_B$ has to be rejected while not enough evidence is available for the rejection of the hypotheses $H_A$ and $H_C$.

Moving on to Example \ref{cv} we can say that the analysis of the findings with the two different tests appears to be more complex than the previous one, see Table \ref{4fbst}. In case $[A]$, for both BDT and FBST with the flat reference function, there is not enough evidence to reject the hypothesis. On the contrary, if one considers the FBST with the Jeffreys' prior as reference function, one is led to reject this hypothesis. In case [B], by rejecting the hypothesis, the BDT is in agreement with the FBST with the Jeffreys' reference function in contrast to the FBST with the flat reference function for which there is not enough evidence to reject it. 

Finally, in the case illustrated in Example \ref{Mara}, the conclusion reached with the FBST and with the BDT is the same (for both reference functions considered), i.e. there is no enough evidence to reject the hypothesis (see Table \ref{5fbst}). It should be noted that, again, there are substantial differences between the values of the evidence measures. 

\vspace{0.2cm}

\begin{table}[!htbp]
\centering
\begin{tabular}{l|cc|c}
                   & \multicolumn{2}{c|}{\textit{\textbf{e-value}}}                   & \multirow{2}{*}{\textit{$\boldsymbol{\delta_H}$}} \\
\textit{\textbf{}} & \textit{$r(\theta)\propto 1$} & \textit{$r(\theta)=g_0(\theta)$} &                                               \\ \hline
$H_A: \alpha = 2.5$   & 0.557                             & 0.186                                & 0.570                                         \\
$H_B: \mu = 6$    & {0.984}                          & {0.963}                                &  0.976    
\\
$H_C: \sigma^2 = 2$    & {0.784}                          & {0.562}                                & 0.846
\end{tabular}
\vspace{0.4cm}
\caption{The table shows the results of the Example \ref{modelloGAMMA} on the test on the shape parameter, mean and variance of the Gamma distribution. For the $e$-value we have considered, as a reference distribution, both a flat reference function and a Jeffreys' prior.}
\label{3fbst}
\end{table}

\begin{table}[!htbp]
\centering
\begin{tabular}{c|cc|c}
                   & \multicolumn{2}{c|}{\textit{\textbf{e-value}}}                   & \multirow{2}{*}{\textit{$\boldsymbol{\delta_H}$}} \\
\textit{\textbf{}} & \textit{$r(\theta)\propto 1$} & \textit{$r(\theta)=g_0(\theta)$} &                                               \\ \hline
$[A]\ n=10$   & {0.364}                            &            {0.999}                      & 0.570                                         \\
$[B]\ n=40$    & {0.924}                        & {1}                                & 0.972                       
\end{tabular}
\vspace{0.4cm}
\caption{The table shows the results of the Example \ref{cv} on the test of the coefficient of variation for a Normal distribution. For the $e$-value we have considered, as a reference distribution, both a flat reference function and a Jeffreys' prior.}
\label{4fbst}
\end{table}

\begin{table}[!htbp]
\centering
\begin{tabular}{c|cc|c}
                   & \multicolumn{2}{c|}{\textit{\textbf{e-value}}}                   & \multirow{2}{*}{\textit{$\boldsymbol{\delta_H}$}} \\
\textit{\textbf{}} & \textit{$r(\theta)\propto 1$} & \textit{$r(\theta)=g_0(\theta)$} &                                               \\ \hline
$H:  \gamma=2$   & {0.650}                            & {0.691}                                & 0.844                                         
\end{tabular}
\vspace{0.4cm}
\caption{The table shows the results of the Example \ref{Mara} on the test of the skewness coefficient of the Inverse Gaussian distribution. For the $e$-value we have considered, as a reference distribution, both a flat reference function and a Jeffreys' prior.}
\label{5fbst}
\end{table}

The calculation of the FBST for a scalar parameter of interest without nuisance parameters, has been carried out through the function defined in the `fbst' package \cite{kelter2022fbst} for \Rlogo. Instead, tangential sets $\overline{T}$ and its integrals, for Examples \ref{modelloGAMMA}, \ref{cv} and \ref{Mara}, were determined by means of the \textit{Mathematica} software. Browsing through the code that leads to the calculation of these measures (see \cite{maramanca}), it is evident that more work is required for the calculation of the integration region related to the FBST. In this sense, the BDT appears to be easier to apply. 
\end{example}

%------------------------------------ Conclusions -------------------------------------------------------------------------
\section{Conclusions}\label{sec5}
We  propose a new measure of evidence in a  Bayesian perspective.
From an examination of the examples illustrated, the conceptual simplicity of the proposed method is evident as well as its theoretical consistency. We have presented some simple cases where the computation of the BDM is straightforward. 

In some situations,  the BDM can be usefully applied adopting a subjective prior.
It is indeed interesting the situation where  one or more statisticians choose the hypothesis $H$ and the prior according to his or their knowledge. In such cases the BDT  would have a confirmatory value.  The use of subjective priors  must be accompanied by a robustness study especially in the case of small sample sizes.

So far we have considered only hypotheses that induce a partition on the parameter space, but the extension of the definition and the analysis of the BDT to more complex hypotheses is under investigation. Theoretical and computational developments in more general contexts are also being explored.
%---------------------------------------------------------------------------------------------------------------------------

\end{document}